\documentclass[sn-mathphys]{sn-jnl}

\jyear{2021}%

\usepackage{bm}
\usepackage{epstopdf}
\begin{document}

\title[A high-fidelity quantum state transfer algorithm on the complete bipartite graph]{A high-fidelity quantum state transfer algorithm on the complete bipartite graph}
\author*[1]{\fnm{Dan} \sur{Li}}\email{lidansusu007@163.com}
\author[1]{\fnm{Jia-Ni} \sur{Huang}}\email{huangjianiyee@163.com}
\author[1]{\fnm{Yu-Qian} \sur{Zhou}}\email{zhouyuqian@nuaa.edu.cn}
\author[2]{\fnm{Yu-Guang} \sur{Yang}}\email{yangyang7357@bjut.edu.cn}
\affil[1]{\orgdiv{College of Computer Science and Technology}, \orgname{Nanjing University of Aeronautics and Astronautics}, \orgaddress{\city{Nanjing}, \postcode{211106}, \country{China}}}

\affil[2]{\orgdiv{College of Computer Science and Technology}, \orgname{Beijing University of Technology}, \orgaddress{\state{Beijing}, \postcode{100124}, \country{China}}}


\abstract{
	High-fidelity quantum state transfer is critical for quantum communication and scalable quantum computation. 
	Current quantum state transfer algorithms on the complete bipartite graph, which are based on discrete-time quantum walk search algorithms, suffer from low fidelity in some cases.  
	To solve this problem, in this paper we propose a two-stage quantum state transfer algorithm on the complete bipartite graph. 
	The algorithm is achieved by the generalized Grover walk with one marked vertex. 
	The generalized Grover walk's coin operators and the query oracles are both parametric unitary matrices, 
	which are designed flexibly based on the positions of the sender and receiver and the size of the complete bipartite graph. 
	We prove that the fidelity of the algorithm is greater than $1-2\epsilon_{1}-\epsilon_{2}-2\sqrt{2}\sqrt{\epsilon_{1}\epsilon_{2}}$ or $1-(2+2\sqrt{2})\epsilon_{1}-\epsilon_{2}-(2+2\sqrt{2})\sqrt{\epsilon_{1}\epsilon_{2}}$ for any adjustable parameters $\epsilon_{1}$ and $\epsilon_{2}$ when the sender and receiver are in the same partition or different partitions of the complete bipartite graph. 
	The algorithm provides a novel approach to achieve high-fidelity quantum state transfer on the complete bipartite graph in any case, which will offer potential applications for quantum information processing.
}




\maketitle

\section{Introduction}
Quantum walk \cite{kadian2021quantum,venegas2012quantum}, the quantum counterpart of classical random walk, was first proposed by Aharonov \cite{aharonov1993quantum} in 1993. 
It is a universal model of quantum computation\cite{childs2009universal,lovett2010universal} and has become a useful tool for designing quantum algorithms, such as quantum search algorithms \cite{reitzner2009quantum,rhodes2019quantum}, quantum state transfer algorithms \cite{yalccinkaya2015qubit,zhan2014perfect}, quantum hash functions\cite{li2022controlled,li2013discrete}, and so on\cite{ambainis2007quantum,magniez2007quantum,reitzner2017finding}.
There are two kinds of quantum walks, discrete-time quantum walks\cite{yalccinkaya2015qubit,zhan2014perfect} and continuous-time quantum walks\cite{wang2019controlled,childs2004spatial,philipp2016continuous}. 

Quantum state transfer is to transfer the initial state from the sender to the receiver which is critical for quantum communication and scalable quantum computation \cite{divincenzo2000physical}. 
When the fidelity of the quantum state transfer algorithm is $1$, we call it perfect state transfer. 
It can be divided into two cases: the position of the sender and receiver are known or unknown.
When the position of the sender and the receiver vertices are known, we can globally design the dynamics to transfer the walker from one to the other. 
It was investigated on different graphs such as a line \cite{yalccinkaya2015qubit,zhan2014perfect}, a circle \cite{yalccinkaya2015qubit}, a 2D lattice \cite{zhan2014perfect}, a regular graph \cite{shang2019quantum}, a complete graph\cite{shang2019quantum} or more general networks \cite{chen2019quantum}.
When the position of the sender and the receiver are unknown, the Grover walk with two marked vertices, the sender and receiver, is used to achieve state transfer. 
It was analyzed on various types of graphs such as a star graph\cite{vstefavnak2016perfect}, a complete bipartite graph \cite{vstefavnak2017perfect}, a complete M-partite graph \cite{skoupy2021quantum}, or a circulant graph \cite{zhan2019infinite}. In this paper, we consider the latter.

Current quantum state transfer algorithms on the complete bipartite graph, which are based on discrete-time quantum walk search algorithms, have low fidelity in some cases. 
Ref. \cite{vstefavnak2017perfect} has proved that perfect state transfer can not be achieved when the sender and receiver are in opposite partitions with different sizes. The fidelity is low when the number of vertices in the two partitions differs greatly. Ref. \cite{santos2022quantum} uses lackadaisical quantum walks to achieve state transfer. But it achieves high fidelity only when the number of vertices in two partitions of the complete bipartite graph exceeds a certain number. 

To avoid the problem of low fidelity, in this paper we propose a two-stage quantum state transfer algorithm on the complete bipartite graph that achieves high-fidelity quantum state transfer in any case.
It is inspired by Ref. \cite{xu2022robust}. 
As shown in Fig. \ref{model}, the initial state is transferred to the uniform superposition state of the vertices on the other side of the sender with the fidelity of at least $1-\epsilon_{1}$ in the first stage. 
In the second stage, the uniform superposition state of the vertices on the other side of the sender is transferred to the target state with the fidelity of at least $1-\epsilon_{2}$, when the sender and receiver are in the same partition or different partitions. 

The algorithm is achieved by the generalized Grover walks with one marked vertex. 
In the first stage, the marked vertex is the sender. But in the second stage, the receiver is the marked vertex.
The coin operators of the generalized Grover walk and the query oracles are both parametric unitary matrices changed with time which are designed according to the position of the sender and receiver and the size of the complete bipartite graph. 

Through analysis, it is found that the fidelity of the quantum state transfer algorithm is greater than $ 1-2\epsilon_{1}-\epsilon_{2}-2\sqrt{2}\sqrt{\epsilon_{1}\epsilon_{2}}$ or $1-(2+2\sqrt{2})\epsilon_{1}-\epsilon_{2}-(2+2\sqrt{2})\sqrt{\epsilon_{1}\epsilon_{2}}$ when the sender and receiver are in the same partition or different partitions. $\epsilon_{1}, \epsilon_{2}$ are tunable parameters chosen from (0,1]. 
When $\epsilon_{1}$ and $\epsilon_{2}$ are small, the value of fidelity of the quantum state transfer algorithm will be close to 1. 
The advantage of the algorithm is it works in any case since high-fidelity quantum state transfer can be reached by adjusting the parameters of the coin operators and the query oracles.
\begin{figure}[H]
	\centering
	\includegraphics[width=1\textwidth]{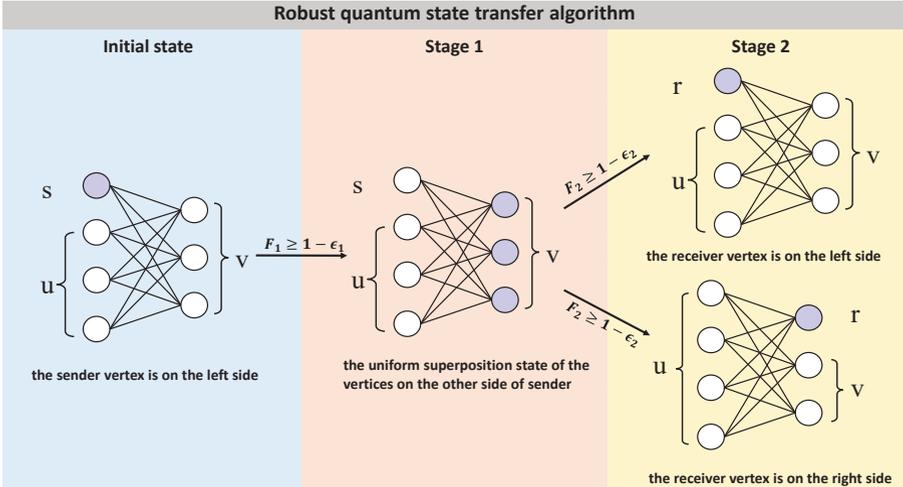}
	\caption{The process of the quantum state transfer algorithm. }
	\label{model}
\end{figure}
The rest of this paper is organized as follows. In section \ref{sec2}, some preliminaries are introduced. The quantum state transfer algorithm is presented in section \ref{sec3} and section \ref{sec4}. A conclusion is presented in Section \ref{sec5}.

\section{Preliminaries}\label{sec2}
\textbf{Complete bipartite graph. }
Let $G=(V, E)$ be a graph where $V$ is the vertex set and $E$ is the edge set. For $u \in V$, $deg(u)=\{ v \vert (u,v)\in E \} $ denotes the set of neighbors of $u$. The degree of $u$ is denoted as $d_{u}=\vert deg(u)\vert $.
A bipartite graph can be denoted as $G=\{V_{1} \cup V_{2},E=\{(u,v)\vert u\in V_{1},v\in V_{2}\}\}$ with $V_{1} \cap V_{2}= \varnothing$.
$V_{1}$ and $V_{2}$ denote the vertices on the left side of the bipartite graph and the right side of it respectively.
A complete bipartite graph is a bipartite graph where every vertex on the left side is connected to every vertex on the right side. A complete bipartite graph is shown in Fig. \ref{graphtwo}, which contains 4 vertices on the left side and 3 vertices on the right side.
\begin{figure}[H]
	\centering
	\includegraphics[width=0.4\textwidth]{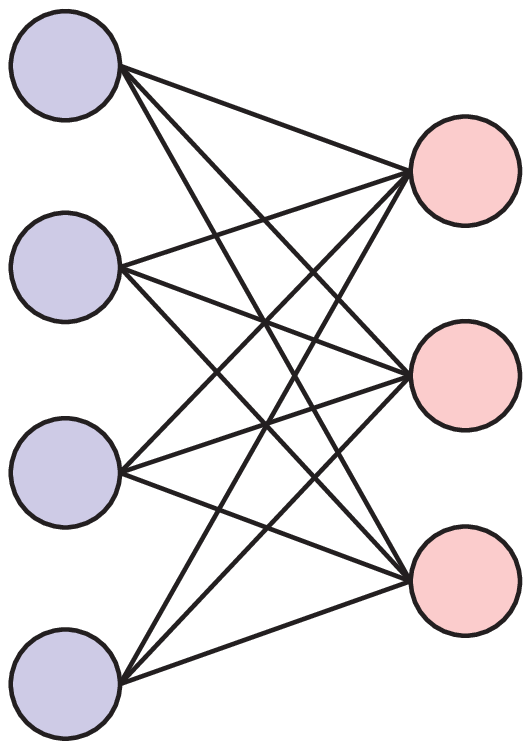}
	\caption{A complete bipartite graph with 4 vertices on the left side and 3 vertices on the right side.}
	\label{graphtwo}
\end{figure}
\textbf{Generalized Grover walk. }
A coined walk is called the Grover walk if the coin operator is the Grover matrix.
The Grover walk is generalized by considering coin operators as parametric unitary matrices, which include the Grover matrix as a special case for some particular values of the parameters.

The Hilbert space of generalized Grover walk on a graph $G = (V, E)$ can be defined as
\begin{equation}
	H^{N^2 }=span\{\vert uv\rangle ,(u,v)\in E\}, 
\end{equation}
where $u$ is the position of the walker and $v$ is the coin that represents a neighbor of $u$. $N$ is the number of vertices in the complete bipartite graph. 

The evolution operator of the generalized Grover walk with marked vertex used in this paper is denoted as 
\begin{equation}
	U(\alpha,\beta)=SC(\alpha)Q(\beta), 
\end{equation}
where the flip-flop shift operator $S$ is
\begin{equation}
	S=\sum_{u,v}\vert u,v\rangle \langle v,u\vert. 
\end{equation}
The coin operator $C(\alpha)$ is 
\begin{equation}
	C(\alpha)=I\otimes\sum\limits_{u}[(1-e^{-i\alpha})\vert \Psi_{u}\rangle \langle\Psi_{u}\vert -I], 
\end{equation}
where
\begin{equation}
	\vert \Psi_{u}\rangle =\frac{1}{\sqrt{d_{u}}}\sum\limits_{v\in deg(u)}\vert v\rangle. 
\end{equation}
The query oracle $Q(\beta)$ is 
\begin{equation}
	Q(\beta)\vert uv\rangle=\left\{
	\begin{aligned}
		e^{i\beta}\vert uv\rangle & , & when \  u \  is  \  marked,\\
		\vert uv\rangle & , &  when  \  u  \  is  \  not \ marked.
	\end{aligned}
	\right. 
\end{equation}
Let the initial state be $\vert \psi_{0}\rangle$. The state of the walker after $t$ steps is given by
\begin{equation}
	\vert \psi_{t}\rangle=U(\alpha_{t},\beta_{t})U(\alpha_{t-1},\beta_{t-1})...U(\alpha_{2},\beta_{2})U(\alpha_{1},\beta_{1})\vert \psi_{0}\rangle.
\end{equation}
\textbf{Quantum state transfer. }
The initial state of the quantum state transfer algorithm is
\begin{equation}
	\vert \psi_{0}\rangle=\frac{1}{\sqrt{d_{s}}}\sum\limits_{v\in deg(s)}\vert sv\rangle ,
\end{equation}
where $s$ is the position of the sender. 

The target state of quantum state transfer is
\begin{equation}
	\vert target\rangle=\frac{1}{\sqrt{d_{r}}}\sum\limits_{v\in deg(r)}\vert rv\rangle, 
\end{equation}
where $r$ is the position of the receiver.
The fidelity of the final state and the target state is given by
\begin{equation}
	F(t)=\vert \langle target\vert \psi_{t}\rangle\vert ^{2}. 
\end{equation}
We call it perfect state transfer when the value of fidelity is 1.

\textbf{Quasi-Chebyshev polynomial. } The Chebyshev polynomials of the first kind $T_{n}(x)$ are defined by initial values $T_{0}(x)=1$, $T_{1}(x)=x$, and for an integer $n \ge 2$,
\begin{equation}
	T_{n}(x)=2xT_{n-1}(x)-T_{n-2}(x).
\end{equation}
A result of one Quasi-Chebyshev polynomial implied in \cite{yoder2014fixed} is stated in the following lemma.
\newtheorem{lem}{Lemma}
\begin{lem}\label{lem1}
	 Let $x=cos(\theta)$ for $\theta \in [0,2\pi]$. Let $h\ge3 $ be an odd integer.
	 One Quasi-Chebyshev polynomial $a_k(x)$ is defined by initial values $a_{0}(x)=1,a_1(x)=x$, and for $k=2,...,h,$
	 \begin{equation}
	 	a_{k}(x)=x(1+e^{-i(\eta_{k}-\eta_{k-1})})a_{k-1}(x)-e^{-i(\eta_{k}-\eta_{k-1})}a_{k-2}(x) .
	 \end{equation}
	 When $\eta_{k+1}-\eta_{k}=(-1)^{k}\pi-2arccot(tan(\frac{k\pi}{h})\sqrt{1-\gamma^{2}}) $ for $k=1,...,h-1,$ where $
	 \gamma=\frac{1}{cos(\frac{1}{h}arccos(\frac{1}{\sqrt{\epsilon}}))} $ with $ \epsilon \in (0,1],$ we have $
	 a_h(x)=\frac{T_h(\frac{x}{\gamma})}{T_h(\frac{1}{\gamma})} $ with $ T_h(\frac{1}{\gamma})=\frac{1}{\sqrt{\epsilon}}.
	 $
\end{lem}
\section{Sender and receiver in the same partition}\label{sec3}
The quantum state transfer algorithm will be different when the sender and receiver are in the same partition or different partitions. In this section, we propose a quantum state transfer algorithm when the sender and receiver are in the same partition. As shown in Fig. \ref{first_case}, the sender and the receiver are on the left side of the complete bipartite graph. 
The left side of the complete bipartite graph has $m$ vertices and the right side of it has $n$ vertices.
\begin{figure}[H]
	\centering
	\includegraphics[width=0.4\textwidth]{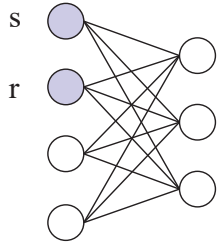}
	\caption{The sender and the receiver are on the left side of the complete bipartite graph.}
	\label{first_case}
\end{figure}
Our algorithm is as follows.
\begin{algorithm}[H]
	\caption{Quantum state transfer algorithm (sender and receiver in the same partition)}\label{algo1}
	\renewcommand{\algorithmicrequire}{\textbf{Input:}}
	\begin{algorithmic}
		\Require the initial state$\vert\psi_{0}\rangle$, parameters $\epsilon_{1}$ and $\epsilon_{2}$. 
	\end{algorithmic}
	\textbf{First stage:}
	\begin{algorithmic}
		\State \textbf{Initialization: }
		\State {\indent}Let $h_{1}$ be an odd integer and ensure $h_{1} \ge ln(\frac{2}{\sqrt{\epsilon_{1}}})\sqrt{m}$.
		\State {\indent}Let $\beta_{k}=-\alpha_{h_{1}+2-k}=-2arccot(tan(\frac{(k-1)\pi}{h_{1}})\sqrt{1-\gamma_{1}^2})$ for $k=3,5,7,...,h_{1}$, where $\gamma_{1}=\frac{1}{cos(\frac{1}{h_{1}}arccos(\frac{1}{\sqrt{\epsilon_{1}}}))}$. The other $\alpha_{i}$ and $\beta_{i},$ can be any value. 
		\State \textbf{Perform the evolution operator: }
		\State{\indent}$\vert \psi_{h_{1}}\rangle=U(\alpha_{h_{1}},\beta_{h_{1}})U(\alpha_{h_{1}-1},\beta_{h_{1}-1})...U(\alpha_{2},\beta_{2})U(\alpha_{1},\beta_{1})\vert\psi_{0}\rangle$
	\end{algorithmic}
	\textbf{Second stage:}
	\begin{algorithmic}
		\State \textbf{Initialization: }
		\State {\indent}Let $h_{2}$ be an odd integer and ensure $h_{2} \ge ln(\frac{2}{\sqrt{\epsilon_{2}}})\sqrt{m}$.
		\State {\indent}Let $\alpha^{'}_{k}=-\beta^{'}_{h_{2}+2-k}=2arccot(tan(\frac{(k-1)\pi}{h_{2}})\sqrt{1-\gamma_{2}^2})$
		for
		$k=3,5,7,...,h_{2}$, where $\gamma_{2}=\frac{1}{cos(\frac{1}{h_{2}}arccos(\frac{1}{\sqrt{\epsilon_{2}}}))}$. The other $\alpha_{i}^{'}$ and $\beta_{i}^{'}$ can be any value.
		\State \textbf{Perform the evolution operator: }
		\State {\indent }$\vert \psi_{h_{2}}\rangle=U(\alpha^{'}_{h_{2}},\beta^{'}_{h_{2}})U(\alpha^{'}_{h_{2}-1},\beta^{'}_{h_{2}-1})...U(\alpha^{'}_{2},\beta^{'}_{2})U(\alpha^{'}_{1},\beta^{'}_{1})\vert\psi_{h_{1}}\rangle$
	\end{algorithmic}
\end{algorithm}
Our algorithm is divided into two stages. The purpose of the first stage is to transfer the initial state to the uniform superposition state of the vertices on the other side of the sender. In the first stage, only the sender is the marked vertex. The purpose of the second stage is to transfer the uniform superposition state of the vertices on the other side of the sender to the target state. In the second stage, only the receiver is the marked vertex.

The analysis of the first stage and the second stage are shown in \ref{3.1} and \ref{3.2} respectively.
The analysis of the fidelity of the quantum state transfer algorithm is shown in \ref{3.3}. 
\subsection{The first stage of the quantum state transfer algorithm}\label{3.1}
In the first stage, only the sender is marked. Thus, the vertices can be divided into three parts shown in Fig. \ref{first_stage}: the sender denoted by $s$ on the left side, the other vertices denoted by $u$ on the left side, and $v$ on the right side. 
Therefore, the analysis can be simplified in an invariant subspace with the orthogonal
basis $\{\vert e_{1}\rangle, \vert e_{2}\rangle, \vert e_{3}\rangle, \vert e_{4}\rangle \}$ given below. The orthogonal basis is only used in \ref{3.1}. 
\begin{equation}
	\begin{aligned}
		&\vert e_{1}\rangle=\frac{1}{\sqrt{n}}\sum\limits_{v}\vert sv\rangle
 		,\\
		&\vert e_{2}\rangle=\frac{1}{\sqrt{n}}\sum\limits_{v}\vert vs\rangle
		,\\
		&\vert e_{3}\rangle=\frac{1}{\sqrt{n(m-1)}}\sum\limits_{v,u}\vert vu\rangle,\\
		&\vert e_{4}\rangle=\frac{1}{\sqrt{n(m-1)}}\sum\limits_{u,v}\vert uv\rangle. 
	\end{aligned}.
\end{equation}
\begin{figure}[H]
	\centering
	\includegraphics[width=0.4\textwidth]{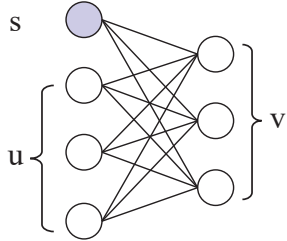}
	\caption{Only the sender is marked in the first stage.}
	\label{first_stage}
\end{figure}
So the initial state can be denoted as
$\vert \psi_0\rangle=\frac{1}{\sqrt{n}}\sum\limits_{v}\vert sv\rangle=\vert e_{1}\rangle=(1,0,0,0)^T$.
The target state of the first stage can be denoted as
$
\vert \Psi\rangle=\frac{1}{\sqrt{mn}}(\sum\limits_{v,u}\vert vu\rangle+\sum\limits_{v}\vert vs\rangle)=\frac{1}{\sqrt{m}}\vert e_{2}\rangle+\frac{\sqrt{m-1}}{\sqrt{m}}\vert e_{3}\rangle=(0,\frac{1}{\sqrt{m}},\frac{\sqrt{m-1}}{\sqrt{m}},0)^T
$. 

The flip-flop shift operator $S$, the query oracle $Q(\beta)$, and the coin operator $C(\alpha)$ can be rewritten as 
\begin{equation}
	S=\begin{pmatrix}
		0&1&0&0\\
		1&0&0&0\\
		0&0&0&1\\
		0&0&1&0\\
	\end{pmatrix}
	,
	\begin{aligned}
		Q(\beta)=\begin{pmatrix}
			e^{i\beta}&0&0&0\\
			0&1&0&0\\
			0&0&1&0\\
			0&0&0&1\\
		\end{pmatrix}
	\end{aligned}
	,
\end{equation}
and
\begin{equation}
	\begin{aligned}
		C(\alpha)=\begin{pmatrix}
			-e^{-i\alpha}&0&0&0\\
			0&\frac{(1-e^{-i\alpha})(1-cos(\omega))}{2}-1&\frac{(1-e^{-i\alpha})sin(\omega)}{2}&0\\
			0&\frac{(1-e^{-i\alpha})sin(\omega)}{2}&\frac{(1-e^{-i\alpha})(1+cos(\omega))}{2}-1&0\\
			0&0&0&-e^{-i\alpha}\\
		\end{pmatrix}
	\end{aligned}
	,
\end{equation}
where $cos(\omega)=1-\frac{2}{m},sin(\omega)=\frac{2}{m}\sqrt{m-1}$.

In the first stage, we know  
\begin{equation}
	\begin{aligned}
		\vert \psi_{h_{1}}\rangle &=SC(\alpha_{h_{1}})Q(\beta_{h_{1}})SC(\alpha_{h_{1}-1})Q(\beta_{h_{1}-1})...SC(\alpha_2)Q(\beta_2)SC(\alpha_1)Q(\beta_1)\vert \psi_0\rangle.\\
	\end{aligned}
\end{equation}
The coin operator $C(\alpha)$ can be denoted as 
\begin{equation}
	C(\alpha)=e^{-\frac{i\alpha}{2}}A(\frac{\pi}{2})R(\alpha)A(-\frac{\pi}{2})
	\label{C}, 
\end{equation}
where 
\begin{equation}
	R(\theta)=-\begin{pmatrix}
		e^{-\frac{i\theta}{2}}&0&0&0\\
		0&e^{\frac{i\theta}{2}}&0&0\\
		0&0&e^{-\frac{i\theta}{2}}&0\\
		0&0&0&e^{-\frac{i\theta}{2}}\\
	\end{pmatrix}, 
\end{equation}
and 
\begin{equation}
	A(\theta)=\begin{pmatrix}
		1&0&0&0\\
		0&cos(\frac{\omega}{2})&-ie^{i\theta}sin(\frac{\omega}{2})&0\\
		0&-ie^{-i\theta}sin(\frac{\omega}{2})&cos(\frac{\omega}{2})&0\\
		0&0&0&1\\
	\end{pmatrix}.
\end{equation}
The query oracle $Q(\beta)$ can be denoted as 
\begin{equation}
	Q(\beta)=-e^{\frac{i\beta}{2}}SR({\beta})S.
	\label{Q}
\end{equation}
And we find the equation 
\begin{equation}
	SB_1SB_2S=B_2SB_1
	\label{S},
\end{equation}
where $B_1=\prod_{i=0}^{n_1}D_i $ and $B_2=\prod_{i=0}^{n_2}D_i$ for $ D_i\in{A(\theta_i),R(\theta_i)}$. 

By using Eq. (\ref{C}), Eq. (\ref{Q}), and Eq. (\ref{S}), we obtain
\begin{equation}\label{eq1}
	\begin{aligned}
		\vert \psi_{h_{1}}\rangle
		&\sim
		R(\beta_{h_{1}})A(\frac{\pi}{2})R(\alpha_{h_{1}-1})A(-\frac{\pi}{2})...
		R(\beta_3)A(\frac{\pi}{2})R(\alpha_{2})A(-\frac{\pi}{2})R(\beta_1)S\\&
		A(\frac{\pi}{2})R(\alpha_{h_{1}})A(-\frac{\pi}{2})  R(\beta_{h_{1}-1})A(\frac{\pi}{2})R(\alpha_{h_{1}-2})A(-\frac{\pi}{2})...
		R(\beta_{2})A(\frac{\pi}{2})R(\alpha_{1})A(-\frac{\pi}{2})
		\vert \psi_0\rangle.
	\end{aligned}
\end{equation}

$R(\theta)$  only adds a coefficient to the $\vert \psi_{0}\rangle$, and $A(\theta)$ makes effect only on the second and third dimensions of the $\vert \psi_{0}\rangle$, so Eq. (\ref{eq1}) can be simplified to 
\begin{equation}
	\vert \psi_{h_{1}}\rangle
	\sim
	R(\beta_{h_{1}})A(\frac{\pi}{2})R(\alpha_{h_{1}-1})A(-\frac{\pi}{2})...
	R(\beta_3)A(\frac{\pi}{2})R(\alpha_{2})A(-\frac{\pi}{2})
	S
	\vert \psi_0\rangle.
\end{equation}
Then using $
A(\alpha+\beta)=R(\beta)A(\alpha)R(-\beta)
$ and $
R(\theta)R(-\theta)=I
$, we obtain
\begin{equation}
	\begin{aligned}
		\vert \psi_{h_{1}}\rangle
		\sim& A(\frac{\pi}{2}+\beta_{h_{1}})A(-\frac{\pi}{2}+\beta_{h_{1}}+\alpha_{h_{1}-1})...A(\frac{\pi}{2}+\beta_{h_{1}}+\alpha_{h_{1}-1}+...+\beta_{3})\\&A(-\frac{\pi}{2}+\beta_{h_{1}}+\alpha_{h_{1}-1}+...+\beta_{3}+\alpha_{2})S\vert \psi_0\rangle.
	\end{aligned}
\end{equation}
The purpose of the first stage is to transfer the state $\vert \psi_{0} \rangle$ to the state $\vert \Psi \rangle$. 
The state $\vert \Psi \rangle$ can be denoted as
$
\vert \Psi \rangle=A(\frac{\pi}{2})S\vert e_4\rangle
$.
So the fidelity of the first stage is 
\begin{equation} 
	\begin{aligned}\label{F1}
		F_{1}
		=&\vert \langle e_{4} \vert SA(-\frac{\pi}{2})A(\frac{\pi}{2}+\beta_{h_{1}})A(-\frac{\pi}{2}+\beta_{h_{1}}+\alpha_{h_{1}-1})...A(\frac{\pi}{2}+\beta_{h_{1}}+\alpha_{h_{1}-1}+...+\beta_{3})\\&A(-\frac{\pi}{2}+\beta_{h_{1}}+\alpha_{h_{1}-1}+...+\beta_{3}+\alpha_{2})S\vert \psi_0\rangle\vert^{2}.
	\end{aligned}
\end{equation} 
There exists a set of parameters $\alpha_{i}$, $\beta_{i}$, then the value of fidelity $F_{1}$ will greater than or equal to $1-\epsilon_{1}$. It can be shown in theorem \ref{thm1}.
\newtheorem{thm}{Theorem}
\begin{thm}\label{thm1}
	Let $
	\beta_{k}=-\alpha_{h_{1}+2-k}=-2arccot(tan(\frac{(k-1)\pi}{h_{1}})\sqrt{1-\gamma_{1}^2})
	$ for
	$k=3,5,7...,h_{1},$ where 
	$\gamma_{1}=\frac{1}{cos(\frac{1}{h_{1}}arccos(\frac{1}{\sqrt{\epsilon_{1}}}))}$, and ensure  $
	h_{1} \ge ln(\frac{2}{\sqrt{\epsilon_{1}}})\sqrt{m}
	$, then the value of fidelity $F_{1}$ will be greater than or equal to $1-\epsilon_{1}$. 
\end{thm} 
\begin{proof}[Proof. ]
	Let $
	\beta_{k}=-\alpha_{h_{1}+2-k}
	$ for
	$k=3,5,7...,h_{1}.$ So 
	Eq. (\ref{F1}) can be rewritten as
	\begin{equation} 
		\begin{aligned}\label{F1'}
			F_{1}
			=&\vert \langle e_{4} \vert SA(\phi_{h_{1}})A(\phi_{h_{1}-1})A(\phi_{h_{1}-2})...A(\phi_{2})A(\phi_{1})S\vert \psi_0\rangle\vert^{2},
		\end{aligned}
	\end{equation} 
	where
	$
	\phi_{k+1}-\phi_{k}=-\pi-\beta_{k+1} $ for $ k=2,4,6,...,h_{1}-1
	$ 
	and 
	$
	\phi_{k+1}-\phi_{k}=\pi+\beta_{h_{1}-k+1} $ for $k=1,3,5,...,h_{1}-2
	$.
	
	The formula $
	SA(\phi_{h_{1}})A(\phi_{h_{1}-1})A(\phi_{h_{1}-2})...A(\phi_{2})A(\phi_{1})S\vert \psi_0\rangle
	$ in Eq.(\ref{F1'}) can be viewed as the operator $SA(\phi_{h_{1}})A(\phi_{h_{1}-1})A(\phi_{h_{1}-2})...A(\phi_{2})A(\phi_{1})S$ applied to $\vert \psi_{0} \rangle$. 
	So it can be divided into three steps as follows.
	\begin{equation}
		\vert \psi_{0} \rangle =
		\begin{pmatrix}
			1\\0\\0\\0
		\end{pmatrix}\xrightarrow[\textcircled{1}]{S}
		\begin{pmatrix}
			0\\1\\0\\0
		\end{pmatrix}\xrightarrow[\textcircled{2}]{A(\phi_{h_{1}})A(\phi_{h_{1}-1})A(\phi_{h_{1}-2})...A(\phi_{2})A(\phi_{1})}
		\begin{pmatrix}
			0\\b_{h_{1}}(x)\\c_{h_{1}}(x)\\0
		\end{pmatrix}
		\xrightarrow[\textcircled{3}]{S}
		\begin{pmatrix}
			b_{h_{1}}(x)\\0\\0\\c_{h_{1}}(x)
		\end{pmatrix}\nonumber
	\end{equation}
	
	In step $\textcircled{2}$, the operator $A(\phi_{h_{1}})A(\phi_{h_{1}-1})A(\phi_{h_{1}-2})...A(\phi_{2})A(\phi_{1})$ will be applied to the state $\vert \mu_{0}\rangle=(0,b_0,c_0,0)^T=(0,1,0,0)^T$. Let $\vert \mu_{k}\rangle=(0,b_{k},c_{k},0)^T=A(\phi_{k})\vert \mu_{k-1}\rangle$ for $k=1,2,...,h_{1}$. 
	
	Combined 
	$
	\vert \mu_k\rangle=
	A(\phi_k)\vert \mu_{k-1}\rangle
	$
	and
	$
	\vert \mu_{k-2}\rangle=
	A(\phi_{k-1})^{-1}\vert \mu_{k-1}\rangle
	$, we obtain
	$b_{k}=cos(\frac{\omega}{2})(1+e^{-i(\phi_{k-1}-\phi_{k})})b_{k-1}-e^{-i(\phi_{k-1}-\phi_{k})}b_{k-2}$.
	So the recurrence formula of $b_{k}(x)$ can be defined by 
	$b_0(x)=1,b_1(x)=x$ and for $k=2,3,4,...,h_{1}, $
	\begin{equation}
		b_{k}(x)=x(1+e^{-i(\phi_{k-1}-\phi_{k})})b_{k-1}(x)-e^{-i(\phi_{k-1}-\phi_{k})}b_{k-2}(x),
	\end{equation}
	with $x=cos(\frac{\omega}{2})$.
	
	Let $
	\beta_{k}=-2arccot(tan(\frac{(k-1)\pi}{h_{1}})\sqrt{1-\gamma_{1}^2})
	$ for
	$k=3,5,7...,h_{1},$ where 
	$\gamma_{1}=\frac{1}{cos(\frac{1}{h_{1}}arccos(\frac{1}{\sqrt{\epsilon_{1}}}))}$. 
	So we have $\phi_{k}-\phi_{k+1}=(-1)^{k}\pi-2arccot(tan(\frac{k\pi}{h_{1}})\sqrt{1-\gamma_{1}^2})$ for $k=1,2,...,h_{1}-1$.
	By using lemma \ref{lem1}, 
	we obtain
	\begin{equation}
		b_{h_{1}}(x)=\frac{T_{h_{1}}(\frac{x}{\gamma_{1}})}{T_{h_{1}}(\frac{1}{\gamma_{1}})}
		=\sqrt{\epsilon_{1}}T_{h_{1}}(cos(\frac{1}{h_{1}}arccos(\frac{1}{\sqrt{\epsilon_{1}}}))\sqrt{1-\frac{1}{m}}).
	\end{equation}
	So the fidelity of the first stage can be calculated as follow.
	\begin{equation}
		\begin{aligned}
			F_{1}=1-\vert b_{h_{1}}(x)\vert ^2
			=1-\epsilon_{1}
			T_{h_{1}}^2(cos(\frac{1}{h_{1}}arccos(\frac{1}{\sqrt{\epsilon_{1}}}))\sqrt{1-\frac{1}{m}})
		\end{aligned} 
	\end{equation}
	Let
	$
	h_{1} \ge ln(\frac{2}{\sqrt{\epsilon_{1}}})\sqrt{m}
	$.
	We know $x \ge tanh(x)$ for $x \ge 0$, so we have
	\begin{equation}
		\frac{1}{m} \ge tanh^2(\frac{ln(\frac{2}{\sqrt{\epsilon_{1}}})}{h_{1}}) \textgreater tanh^{2}(\frac{1}{h_{1}}ln(\frac{1}{\sqrt{\epsilon_{1}}}+\sqrt{(\frac{1}{\sqrt{\epsilon_{1}}})^2-1})).
	\end{equation}
	Then using $arccos(z)=\frac{1}{i}ln(z+\sqrt{z^2-1})$ and $tan(iz)=itanh(z)$, we obtain 
	\begin{equation}
		\begin{aligned}
			tanh^{2}(\frac{1}{h_{1}}ln(\frac{1}{\sqrt{\epsilon_{1}}}+\sqrt{(\frac{1}{\sqrt{\epsilon_{1}}})^2-1}))
			&=1-cos^{-2}(\frac{1}{h_{1}}arccos(\frac{1}{\sqrt{\epsilon_{1}}})). \\
		\end{aligned}
	\end{equation}
	So we have
	$
	\frac{1}{m} \ \textgreater 1-cos^{-2}(\frac{1}{h_{1}}arccos(\frac{1}{\sqrt{\epsilon_{1}}})).
	$
	That is
	\begin{equation}
		cos(\frac{1}{h_{1}}arccos(\frac{1}{\sqrt{\epsilon_{1}}}))\sqrt{1-\frac{1}{m}}\ \ \textless 1.
	\end{equation}
	Then we can obtain
	$
	F_{1}\ge 1-\epsilon_{1}
	$.
\end{proof}

Therefore, let $
\beta_{k}=-\alpha_{h_{1}+2-k}=-2arccot(tan(\frac{(k-1)\pi}{h_{1}})\sqrt{1-\gamma_{1}^2})
$ for
$k=3,5,7...,h_{1},$ where 
$\gamma_{1}=\frac{1}{cos(\frac{1}{h_{1}}arccos(\frac{1}{\sqrt{\epsilon_{1}}}))}$, and ensure  $
h_{1} \ge ln(\frac{2}{\sqrt{\epsilon_{1}}})\sqrt{m}
$, the initial state will be transferred to the uniform superposition state of the vertices on the other side of the sender with the fidelity of at least $1-\epsilon_{1}$.

\subsection{The second stage of the quantum state transfer algorithm}\label{3.2}
In the second stage, only the receiver is marked (shown in Fig. \ref{receiver}). Thus the analysis can be simplified in an invariant subspace with the orthogonal
basis $\{\vert e_{1}\rangle, \vert e_{2}\rangle, \vert e_{3}\rangle, \vert e_{4}\rangle \}$  given below. The orthogonal basis is only used in \ref{3.2}. 
\begin{equation}
	\begin{aligned}
		&\vert e_{1}\rangle=\frac{1}{\sqrt{n}}\sum\limits_{v}\vert rv\rangle
		,\\
		&\vert e_{2}\rangle=\frac{1}{\sqrt{n}}\sum\limits_{v}\vert vr\rangle
		,\\
		&\vert e_{3}\rangle=\frac{1}{\sqrt{n(m-1)}}\sum\limits_{v,u}\vert vu\rangle,\\
		&\vert e_{4}\rangle=\frac{1}{\sqrt{n(m-1)}}\sum\limits_{u,v}\vert uv\rangle.
	\end{aligned}
\end{equation}
\begin{figure}[H]
	\centering
	\includegraphics[width=0.4\textwidth]{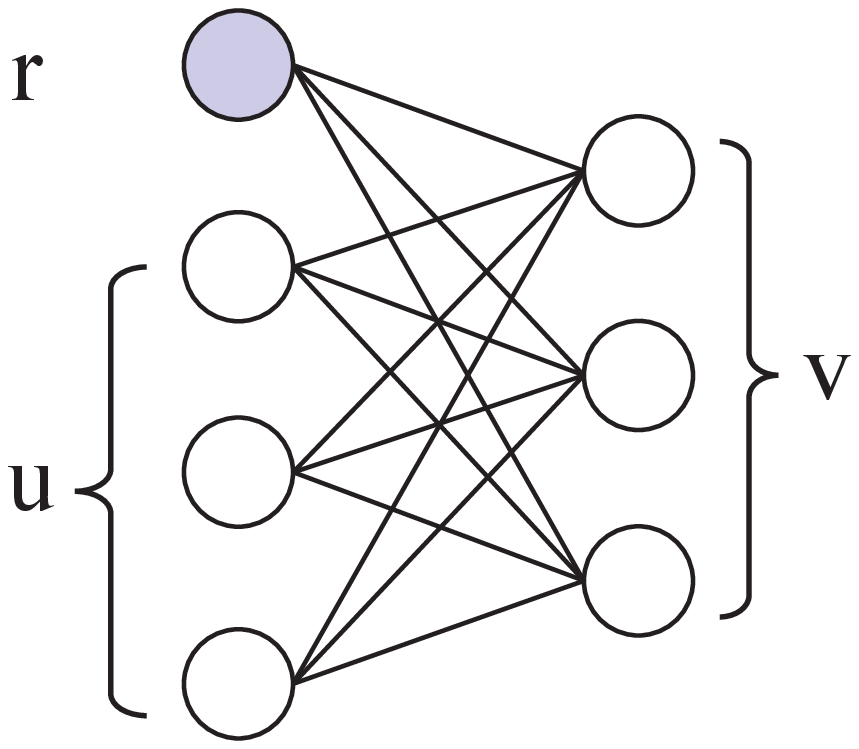}
	\caption{Only the receiver is marked in the second stage.}
	\label{receiver}
\end{figure}
The flip-flop shift operator $S_{1}$, the query oracle $Q_{1}(\beta)$, and the coin operator $C_{1}(\alpha)$ can be rewritten as
\begin{equation}
	\begin{aligned}
		S_1=\begin{pmatrix}
			0&1&0&0\\
			1&0&0&0\\
			0&0&0&1\\
			0&0&1&0\\
		\end{pmatrix}
	\end{aligned},
	\begin{aligned}
		Q_1(\beta)=\begin{pmatrix}
			e^{i\beta}&0&0&0\\
			0&1&0&0\\
			0&0&1&0\\
			0&0&0&1\\
		\end{pmatrix}, 
	\end{aligned}
\end{equation}
and
\begin{equation}
	\begin{aligned}
		C_1(\alpha)=\begin{pmatrix}
			-e^{-i\alpha}&0&0&0\\
			0&\frac{(1-e^{-i\alpha})(1-cos(\omega_1))}{2}-1&\frac{(1-e^{-i\alpha})sin(\omega_1)}{2}&0\\
			0&\frac{(1-e^{-i\alpha})sin(\omega_1)}{2}&\frac{(1-e^{-i\alpha})(1+cos(\omega_1))}{2}-1&0\\
			0&0&0&-e^{-i\alpha}\\
		\end{pmatrix},
	\end{aligned}
\end{equation}
where $cos(\omega_1)=1-\frac{2}{m}$ and $sin(\omega_1)=\frac{2}{m}\sqrt{m-1}$. 

In the second stage, we have
\begin{equation}
	\begin{aligned}
		\vert \psi_{h_{2}}\rangle=S_1C_1(\alpha^{'}_{h_{2}})Q_1(\beta^{'}_{h_{2}})S_1C_1(\alpha^{'}_{h_{2}-1})Q_1(\beta^{'}_{h_{2}-1})...S_1C_1(\alpha^{'}_{1})Q_1(\beta^{'}_{1})\vert \psi_{h_{1}}\rangle. \\
	\end{aligned}
\end{equation}
The coin operator $C_{1}(\alpha)$ can be denoted as 
\begin{equation}
	C_1(\alpha)=e^{-\frac{i\alpha}{2}}A_1(\frac{\pi}{2})R_1(\alpha)A_1(-\frac{\pi}{2})
	\label{C1}, 
\end{equation}
where 
\begin{equation}
	R_1(\theta)=-\begin{pmatrix}
		e^{-\frac{i\theta}{2}}&0&0&0\\
		0&e^{\frac{i\theta}{2}}&0&0\\
		0&0&e^{-\frac{i\theta}{2}}&0\\
		0&0&0&e^{-\frac{i\theta}{2}}\\
	\end{pmatrix}, 
\end{equation}
and 
\begin{equation}
	A_1(\theta)=\begin{pmatrix}
		1&0&0&0\\
		0&cos(\frac{\omega_1}{2})&-ie^{i\theta}sin(\frac{\omega_1}{2})&0\\
		0&-ie^{-i\theta}sin(\frac{\omega_1}{2})&cos(\frac{\omega_1}{2})&0\\
		0&0&0&1\\
	\end{pmatrix}.
\end{equation}
The query oracle $Q_{1}(\beta)$ can be denoted as 
\begin{equation}
	Q_1(\beta)=-e^{\frac{i\beta}{2}}S_1R_1({\beta})S_1. 
	\label{Q1}
\end{equation}
And we find the equation 
\begin{equation}
	S_1B_1S_1B_2S_1=B_2S_1B_1,
	\label{S1}
\end{equation}
where $B_1=\prod_{i=0}^{n_1}D_i,B_2=\prod_{i=0}^{n_2}D_i$ for $D_i\in{A_1(\theta_i),R_1(\theta_i)}$. 

By using Eq. (\ref{C1}), Eq. (\ref{Q1}) and Eq. (\ref{S1}), we obtain 
\begin{equation}\label{eq2}
	\begin{aligned}
		\vert \psi_{h_{2}}\rangle
		\sim&
		R_{1}(\beta^{'}_{h_{2}})A_{1}(\frac{\pi}{2})R_{1}(\alpha^{'}_{h_{2}-1})A_{1}(-\frac{\pi}{2})...
		R_{1}(\beta^{'}_3)A_{1}(\frac{\pi}{2})R_{1}(\alpha^{'}_{2})A_{1}(-\frac{\pi}{2})R_{1}(\beta^{'}_1)S_{1}\\&
		A_{1}(\frac{\pi}{2})R_{1}(\alpha^{'}_{h_{2}})A_{1}(-\frac{\pi}{2})  R_{1}(\beta^{'}_{h_{2}-1})A_{1}(\frac{\pi}{2})R_{1}(\alpha^{'}_{h_{2}-2})A_{1}(-\frac{\pi}{2})...\\
		&
		R_{1}(\beta^{'}_{4})A_{1}(\frac{\pi}{2})R_{1}(\alpha^{'}_{3})A_{1}(-\frac{\pi}{2})R_{1}(\beta^{'}_{2})A_{1}(\frac{\pi}{2})R_{1}(\alpha^{'}_{1})A_{1}(-\frac{\pi}{2})
		\vert \psi_{h_{1}}\rangle.
	\end{aligned}
\end{equation}

The state $\vert \psi_{h_{1}} \rangle$ can be rewritten as 
$
\vert \psi_{h_{1}}\rangle\approx\vert \Psi\rangle=A_1(\frac{\pi}{2})S_1\vert e_4\rangle
$. Then we eliminate invalid $A_{1}(\theta)$ and $R_{1}(\theta)$. So Eq. (\ref{eq2}) can be simplified to 
\begin{equation}
	\begin{aligned}
		\vert \psi_{h_{2}}\rangle
		\sim&
		S_{1}A_{1}(\frac{\pi}{2})R_{1}(\alpha^{'}_{h_{2}})A_{1}(-\frac{\pi}{2})  R_{1}(\beta^{'}_{h_{2}-1})A_{1}(\frac{\pi}{2})R_{1}(\alpha^{'}_{h_{2}-2})A_{1}(-\frac{\pi}{2})...\\
		&
		R_{1}(\beta^{'}_{4})A_{1}(\frac{\pi}{2})R_{1}(\alpha^{'}_{3})A_{1}(-\frac{\pi}{2})R_{1}(\beta^{'}_{2})A_{1}(\frac{\pi}{2})S_{1}
		\vert e_{4}\rangle.
	\end{aligned}
\end{equation}
Then using 
$
A_1(\alpha+\beta)=R_1(\beta)A_1(\alpha)R_1(-\beta)
$ 
and 
$
R_1(\theta)R_1(-\theta)=I
$, we have 
\begin{equation}
	\begin{aligned}
		\vert \psi_{h_{2}}\rangle
		\sim&	S_1A_1(\frac{\pi}{2})A_1(-\frac{\pi}{2}+\alpha^{'}_{h_{2}})A_1(\frac{\pi}{2}+\alpha^{'}_{h_{2}}+\beta^{'}_{h_{2}-1})...\\
		&A_1(-\frac{\pi}{2}+\alpha^{'}_{h_{2}}+\beta^{'}_{h_{2}-1}+...+\beta^{'}_4+\alpha^{'}_{3})
		A_1(\frac{\pi}{2}+\alpha^{'}_{h_{2}}+\beta^{'}_{h_{2}-1}+...+\beta^{'}_2)S_1\vert e_{4}\rangle.
	\end{aligned}
\end{equation}
The target state of the second stage can be denoted as $
\vert target\rangle =\frac{1}{\sqrt{n}}\sum\limits_{v} \vert rv\rangle=\vert e_{1}\rangle=(1,0,0,0)^T
$. 
So the fidelity of the second stage is 
\begin{equation}\label{F2}
	\begin{aligned}
		F_{2}=
		&\vert \langle e_{1}\vert 
		S_1A_1(\frac{\pi}{2})A_1(-\frac{\pi}{2}+\alpha^{'}_{h_{2}})A_1(\frac{\pi}{2}+\alpha^{'}_{h_{2}}+\beta^{'}_{h_{2}-1})...\\
		&A_1(-\frac{\pi}{2}+\alpha^{'}_{h_{2}}+\beta^{'}_{h_{2}-1}+...+\beta^{'}_4+\alpha^{'}_{3})
		A_1(\frac{\pi}{2}+\alpha^{'}_{h_{2}}+\beta^{'}_{h_{2}-1}+...+\beta^{'}_2)S_1\vert e_{4}\rangle \vert ^{2}. 
	\end{aligned}
\end{equation}
There exists a set of parameters $\alpha_{i}^{'}$, $\beta_{i}^{'}$, then the value of fidelity $F_{2}$ will greater than or equal to $1-\epsilon_{2}$. It can be shown in theorem \ref{thm2}.
\begin{thm}\label{thm2}
	Let $\alpha^{'}_{k}=-\beta^{'}_{h_{2}+2-k}=2arccot(tan(\frac{(k-1)\pi}{h_{2}})\sqrt{1-\gamma_{2}^2})$
	for
	$k=3,5,7...,h_{2},$ where 
	$\gamma_{2}=\frac{1}{cos(\frac{1}{h_{2}}arccos(\frac{1}{\sqrt{\epsilon_{2}}}))}$
	, and ensure $
	h_{2} \ge ln(\frac{2}{\sqrt{\epsilon_{2}}})\sqrt{m}
	$, then the value of fidelity $F_{2} \ge 1-\epsilon_{2}$. 
\end{thm}
\begin{proof}[Proof. ]
	Let $\alpha^{'}_{k}=-\beta^{'}_{h_{2}+2-k}$ for $k=3,5,7...,h_{2}$. So Eq. (\ref{F2}) can be rewritten as 
	\begin{equation}\label{F2'''}
		F_{2}=
		\vert \langle e_{1} \vert 
		S_1A_{1}(\eta_{h_{2}})A_{1}(\eta_{h_{2}-1})A_{1}(\eta_{h_{2}-2})...A_{1}(\eta_{2})A_{1}(\eta_{1})S_1\vert e_{4}\rangle\vert ^{2}, 
	\end{equation}
	where $\eta_{k+1}-\eta_{k}=\pi-\alpha^{'}_{k+1}$ for $ k=2,4,6,...,h_{2}-1 $ and $\eta_{k+1}-\eta_{k}=-\pi+\alpha^{'}_{h-k+1} $ for $k=1,3,5,...,h_{2}-2$.
	
	The formula  $S_1A_{1}(\eta_{h_{2}})A_{1}(\eta_{h_{2}-1})A_{1}(\eta_{h_{2}-2})...A_{1}(\eta_{2})A_{1}(\eta_{1})S_1\vert e_{4} \rangle$ in Eq. (\ref{F2'''}) can be viewed as the operator $S_1A_{1}(\eta_{h_{2}})A_{1}(\eta_{h_{2}-1})A_{1}(\eta_{h_{2}-2})...A_{1}(\eta_{2})A_{1}(\eta_{1})S_1$ applied to $\vert e_{4} \rangle $. So it can be divided into three steps as follows.
	\begin{equation}
		\begin{pmatrix}
			0\\0\\0\\1
		\end{pmatrix}\xrightarrow[\textcircled{1}]{S_1}
		\begin{pmatrix}
			0\\0\\1\\0
		\end{pmatrix}\xrightarrow[\textcircled{2}]{A_{1}(\eta_{h_{2}})A_{1}(\eta_{h_{2}-1})A_{1}(\eta_{h_{2}-2})...A_{1}(\eta_{2})A_{1}(\eta_{1})}
		\begin{pmatrix}
			0\\b_{h_{2}}(x)\\c_{h_{2}}(x)\\0
		\end{pmatrix}
		\xrightarrow[\textcircled{3}]{S_1}
		\begin{pmatrix}
			b_{h_{2}}(x)\\0\\0\\c_{h_{2}}(x)
		\end{pmatrix}\nonumber
	\end{equation}
	Then after calculations in step $\textcircled{2}$ like in the proof of theorem \ref{thm1}, the recurrence formula of $c_{k}(x)$ can be defined by $c_0(x)=1,c_1(x)=x$ and for $k=2,3,4,...,h_{2}$,
	\begin{equation}
		c_{k}(x)=x(1+e^{-i(\eta_{k}-\eta_{k-1})})c_{k-1}(x)-e^{-i(\eta_{k}-\eta_{k-1})}c_{k-2}(x),
	\end{equation}
	with $x=cos(\frac{\omega_1}{2})$.
	
	Let $\alpha^{'}_{k}=2arccot(tan(\frac{(k-1)\pi}{h_{2}})\sqrt{1-\gamma_{2}^2})$
	for
	$k=3,5,7...,h_{2},$ where 
	$\gamma_{2}=\frac{1}{cos(\frac{1}{h_{2}}arccos(\frac{1}{\sqrt{\epsilon_{2}}}))}$. Then we get $\eta_{k+1}-\eta_{k}=(-1)^{k}\pi-2arccot(tan(\frac{k\pi}{h})\sqrt{1-\gamma_{2}^2})$.
	By using lemma \ref{lem1}, we obtain
	\begin{equation}
		c_{h_{2}}(x)=\frac{T_{h_{2}}(\frac{x}{\gamma_{2}})}{T_{h_{2}}(\frac{1}{\gamma_{2}})}
		=\sqrt{\epsilon_{2}}T_{h_{2}}(cos(\frac{1}{h_{2}}arccos(\frac{1}{\sqrt{\epsilon_{2}}}))\sqrt{1-\frac{1}{m}}).
	\end{equation}
	So the fidelity of the second stage can be calculated as follow.
	\begin{equation}
		F_2=1-\vert c_{h_{2}}(x)\vert ^2=1-\epsilon_{2} T_{h_{2}}^2(cos(\frac{1}{h_{2}}arccos(\frac{1}{\sqrt{\epsilon_{2}}}))\sqrt{1-\frac{1}{m}})
	\end{equation}
	Let $
	h_{2} \ge ln(\frac{2}{\sqrt{\epsilon_{2}}})\sqrt{m}
	$. Similar to the proof of the theorem \ref{thm1}, we obtain
	$
	F_2\ge 1-\epsilon_{2}
	$.
\end{proof}

Therefore, let $\alpha^{'}_{k}=-\beta^{'}_{h_{2}+2-k}=2arccot(tan(\frac{(k-1)\pi}{h_{2}})\sqrt{1-\gamma_{2}^2})$
for
$k=3,5,7...,h_{2},$ where 
$\gamma_{2}=\frac{1}{cos(\frac{1}{h_{2}}arccos(\frac{1}{\sqrt{\epsilon_{2}}}))}$
, and ensure $
h_{2} \ge ln(\frac{2}{\sqrt {\epsilon_{2}}})\sqrt{m}
$, the uniform superposition state of the vertices on the other side of the sender will be transferred to the target state with the fidelity of at least $1-\epsilon_{2}$.

\subsection{The fidelity of the quantum state transfer algorithm}\label{3.3}
Since the sender and receiver are in the same partition of the complete bipartite graph (shown in Fig. \ref{senderAndReceiver}), the analysis of the algorithm can be simplified in an invariant subspace with the orthogonal
basis $ \{\vert e_{1}\rangle, \vert e_{2}\rangle, \vert e_{3}\rangle, \vert e_{4}\rangle, \vert e_{5}\rangle, \vert e_{6}\rangle \}$ given below. the orthogonal basis is only used in \ref{3.3}.
\begin{equation}
	\begin{aligned}
		&\vert e_{1}\rangle=\frac{1}{\sqrt{n}}\sum\limits_{v}\vert sv\rangle
		,\\
		&\vert e_{2}\rangle=\frac{1}{\sqrt{n}}\sum\limits_{v}\vert vs\rangle
		,\\
		&\vert e_{3}\rangle=\frac{1}{\sqrt{n}}\sum\limits_{v}\vert rv\rangle
		,\\
		&\vert e_{4}\rangle=\frac{1}{\sqrt{n}}\sum\limits_{v}\vert vr\rangle
		,\\
		&\vert e_{5}\rangle=\frac{1}{\sqrt{n(m-2)}}\sum\limits_{u,v}\vert uv\rangle,\\
		&\vert e_{6}\rangle=\frac{1}{\sqrt{n(m-2)}}\sum\limits_{v,u}\vert vu\rangle. 
	\end{aligned}.
\end{equation}
\begin{figure}[H]
	\centering
	\includegraphics[width=0.4\textwidth]{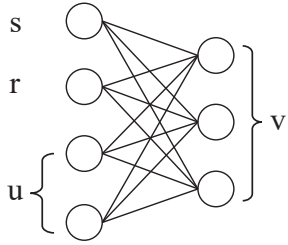}
	\caption{The sender and the receiver are in the same partition. }
	\label{senderAndReceiver}
\end{figure}
From the analysis of the first stage in \ref{3.1}, we know 
$SA(\phi_{h_{1}})A(\phi_{h_{1}-1})...A(\phi_{2})A(\phi_{1})S\vert \psi_0\rangle=(b_{h_{1}}(x),0,0,c_{h_{1}}(x))^{T}$
where $\phi_{h_{1}}=-\frac{\pi}{2}$. And we know $\vert \psi_{h_{1}}\rangle \sim A(\phi_{h_{1}-1})A(\phi_{h_{1}-2})...A(\phi_{2})A(\phi_{1})S\vert \psi_0\rangle$. So we can obtain 
$\vert \psi_{h_{1}}\rangle \sim (0,\sqrt{\frac{m-1}{m}}b_{h_{1}}(x)+\frac{1}{\sqrt{m}}c_{h_{1}}(x),-\frac{1}{\sqrt{m}}b_{h_{1}}(x)+\sqrt{\frac{m-1}{m}}c_{h_{1}}(x),0)^T$.

So in the new basis, the state $\vert \psi_{h_{1}} \rangle$ can be rewritten as 
\begin{equation}
	\vert \psi_{h_{1}}\rangle \sim t_{1}\vert \Psi\rangle+t_{2}\vert e_{2}\rangle,
\end{equation}
where $t_{1}=c_{h_{1}}(x)-\frac{1}{\sqrt{m-1}}b_{h_{1}}(x)$,  $t_{2}=\frac{\sqrt{m}}{\sqrt{m-1}}b_{h_{1}}(x)$,
 and $\vert \Psi \rangle =\frac{1}{\sqrt{m}}\vert e_{2}\rangle +\frac{1}{\sqrt{m}}\vert e_{4}\rangle +\frac{\sqrt{m-2}}{\sqrt{m}}\vert e_{6}\rangle.$ $\vert \Psi \rangle$ denotes the target state of the first stage. 

In the second stage, we have 
\begin{equation}
	\vert \psi_{h_{2}}\rangle \sim t_{1}U_{2}\vert \Psi\rangle+t_{2}U_{2}\vert e_{2}\rangle,
\end{equation}
where $U_{2}$ denotes the evolution operators of the second stage. 

Let $t_{1}U_{2}\vert \Psi\rangle=(f_{1},0,f_{3},0,f_{5},0)^{T}$ and $t_{2}U_{2}\vert e_{2}\rangle=(g_{1},0,g_{3},0,g_{5},0)^{T}$, where $\vert f_{1} \vert^{2} +\vert f_{3} \vert^{2}+\vert f_{5} \vert^{2}=\vert t_{1} \vert ^{2}$ and $\vert g_{1} \vert^{2} +\vert g_{3} \vert^{2}+\vert g_{5} \vert^{2}=\vert t_{2} \vert ^{2}.$ And we can obtain the following equation.

\begin{equation}\label{60}
	\vert f_{1}+g_{1}\vert^{2}+\vert f_{3}+g_{3}\vert^{2}+\vert f_{5}+g_{5}\vert^{2}=1
\end{equation}

The target state of the algorithm is $\vert e_{3}\rangle. $
So the fidelity of the algorithm can be denoted as 
\begin{equation}\label{61}
	F=\bm{\big\vert} f_{3} + g_{3}\bm{\big\vert} ^{2}.
\end{equation}

From Eq. (\ref{60}) and Eq. (\ref{61}), we can obtain 
\begin{equation}\label{62}
	F=1-\bm{\big\vert} f_{1}+g_{1}\bm{\big\vert} ^{2}-\bm{\big\vert} f_{5}+g_{5}\bm{\big\vert} ^{2}.
\end{equation}
By using $\vert x+y \vert \leq \vert \vert x \vert +\vert y \vert \vert $, we can obtain
\begin{equation}
	F\geq 1-\bm{\big\vert} f_{1}\bm{\big\vert} ^{2}-\bm{\big\vert} g_{1}\bm{\big\vert}^{2}-2\bm{\big\vert} f_{1}\bm{\big\vert} \bm{\big\vert} g_{1}\bm{\big\vert}
	-\bm{\big\vert} f_{5}\bm{\big\vert} ^{2}-\bm{\big\vert} g_{5}\bm{\big\vert}^{2}-2\bm{\big\vert} f_{5}\bm{\big\vert} \bm{\big\vert} g_{5}\bm{\big\vert} 
	.
\end{equation}
From \ref{3.2}, 
we know $\bm{\big\vert} f_{1}\bm{\big\vert} ^{2}+\bm{\big\vert} f_{5}\bm{\big\vert} ^{2}\textless \epsilon_{2}$. 
And we know $\bm{\big\vert} g_{1}\bm{\big\vert} ^{2}+\bm{\big\vert} g_{5}\bm{\big\vert} ^{2}\leq \vert t_{2} \vert^{2} \leq 2\epsilon_{1}$. Then we obtain $\bm{\big\vert} f_{1}\bm{\big\vert} 
\bm{\big\vert} g_{1}\bm{\big\vert}+
\bm{\big\vert} f_{5}\bm{\big\vert} 
\bm{\big\vert} g_{5}\bm{\big\vert} 
\leq \sqrt{(\bm{\big\vert} f_{1}\bm{\big\vert}^2+\bm{\big\vert} f_{5}\bm{\big\vert}^2)
	(\bm{\big\vert} g_{1}\bm{\big\vert}^2+\bm{\big\vert} g_{5}\bm{\big\vert}^2)
}\textless \sqrt{2\epsilon_{1}\epsilon_{2}}$. So we have
\begin{equation}\label{lowerbound}
	F \textgreater 1-\epsilon_{2}-2\epsilon_{1}-2\sqrt{2}\sqrt{\epsilon_{1}\epsilon_{2}}.
\end{equation}
From Eq. (\ref{lowerbound}), we know that the fidelity will be close to 1 when $\epsilon_{1}$ and $\epsilon_{2}$ are small.
For instance, let $\epsilon_{1}$ be $0.01$ and $\epsilon_{2}$ be $0.01$. From Eq. (\ref{lowerbound}), we know the fidelity will be greater than $0.94$ regardless of the value of $m$ and $n$.
The simulation results of the algorithm are shown in Fig. \ref{fidelity}. The fidelity is bigger than 0.98 at a certain range of $m$ and $n$ when $\epsilon_{1}=0.01$ and $\epsilon_{2}=0.01$. 
It further verifies that the quantum state transfer algorithm can achieve high fidelity. 
\begin{figure}[H]
	\centering
	\includegraphics[width=0.6\textwidth]{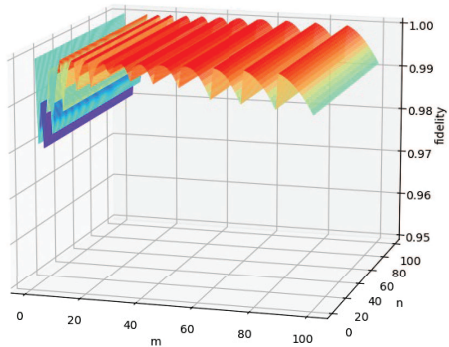}
	\caption{The fidelity of the quantum state transfer algorithm with $\epsilon_{1}=0.01$, $\epsilon_{2}=0.01$ when the sender and receiver are in the same partition. }
	\label{fidelity}
\end{figure}
\section{Sender and receiver in different partitions}\label{sec4}
In this section, we propose the quantum state transfer algorithm when the sender and receiver are in different partitions. As shown in Fig. \ref{second_case}, the sender is on the left side of the complete bipartite graph and the receiver is on the right side of it. The left side of the complete bipartite graph has $m$ vertices and the right side of it has $n$ vertices.
\begin{figure}[H]
	\centering
	\includegraphics[width=0.4\textwidth]{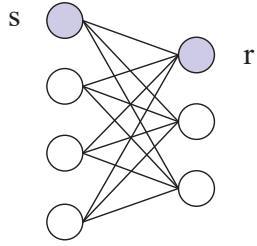}
	\caption{The sender is on the left side of the complete bipartite graph and the receiver is on the right side of it.}
	\label{second_case}
\end{figure}
Our algorithm is as follows.
\begin{algorithm}[H]
	\caption{Quantum state transfer algorithm (the sender and receiver in different partitions) }\label{algo1}
	\renewcommand{\algorithmicrequire}{\textbf{Input:}}
	\begin{algorithmic}
		\Require the initial state$\vert\psi_{0}\rangle$, parameters $\epsilon_{1}$ and $\epsilon_{2}$.
	\end{algorithmic}
	\textbf{First stage:}
	\begin{algorithmic}
		\State \textbf{Initialization: }
		\State {\indent}Let $h_{1}$ be an odd integer and ensure $h_{1} \ge ln(\frac{2}{\sqrt{\epsilon_{1}}})\sqrt{m}$.
		\State {\indent}Let $\beta_{k}=-\alpha_{h_{1}+2-k}=-2arccot(tan(\frac{(k-1)\pi}{h_{1}})\sqrt{1-\gamma_{1}^2})
		$ for
		$k=3,5,7,...,h_{1}$, where $\gamma_{1}=\frac{1}{cos(\frac{1}{h_{1}}arccos(\frac{1}{\sqrt{\epsilon_{1}}}))}$. The other $\alpha_{i}$ and $\beta_{i}$ can be any value.
		\State \textbf{Perform the evolution operators: }
		\State{\indent $\vert \psi_{h_{1}}\rangle=U(\alpha_{h_{1}},\beta_{h_{1}})U(\alpha_{h_{1}-1},\beta_{h_{1}-1})...U(\alpha_{2},\beta_{2})U(\alpha_{1},\beta_{1})\vert\psi_{0}\rangle$}
	\end{algorithmic}
	\textbf{Second stage:} 
	\begin{algorithmic}
		\State \textbf{Initialization: }
		\State {\indent}Let $h_{2}$ be an even integer and ensure $h_{2} \ge ln(\frac{2}{\sqrt{\epsilon_{2}}})\sqrt{n}$.
		\State {\indent}Let $
		\alpha^{'}_{k}=-\beta^{'}_{h_2+2-k}=2arccot(tan(\frac{k\pi}{h_2+1})\sqrt{1-\gamma_2^2})
		$ for
		$k=2,4,6,...,h_{2}$, where $\gamma_2=\frac{1}{cos(\frac{1}{(h_2+1)}arccos(\frac{1}{\sqrt{\epsilon_2}}))}$. 
		The other $\alpha_{i}^{'}$ and $\beta_{i}^{'}$ can be any value.
		\State \textbf{Perform the evolution operators: }\State{\indent $\vert \psi_{h_{2}}\rangle=U(\alpha^{'}_{h_{2}},\beta^{'}_{h_{2}})U(\alpha^{'}_{h_{2}-1},\beta^{'}_{h_{2}-1})...U(\alpha^{'}_{2},\beta^{'}_{2})U(\alpha^{'}_{1},\beta^{'}_{1})\vert\psi_{h_{1}}\rangle$}
	\end{algorithmic}
\end{algorithm}
Our algorithm is divided into two stages. The purpose of the first stage is to transfer the initial state to the uniform superposition state of the vertices on the other side of the sender. In the first stage, only the sender is the marked vertex. And the second stage is to transfer the uniform superposition state of the vertices on the other side of the sender to the target state. In the second stage, only the receiver is the marked vertex.

The analysis of the first stage and the second stage are shown in \ref{4.1} and \ref{4.2} respectively.
The analysis of the fidelity of the quantum state transfer algorithm is shown in \ref{4.3}. 
\subsection{The first stage of the quantum state transfer algorithm}\label{4.1}
The first stage of the quantum state transfer algorithm is the same when the sender and receiver are in the same partition or different partitions. Therefore, the analysis of the first stage can be viewed in section \ref{3.1}. 
\subsection{The second stage of the quantum state transfer algorithm}\label{4.2}
In the second stage, only the receiver is marked (shown in Fig. \ref{receiver1}). Thus the analysis can be simplified in an invariant subspace with the orthogonal
basis $\{\vert e_{1}\rangle, \vert e_{2}\rangle, \vert e_{3}\rangle, \vert e_{4}\rangle \}$ given below. The orthogonal basis is only used in \ref{4.2}. 
\begin{equation}
	\begin{aligned}
		&\vert e_{1}\rangle=\frac{1}{\sqrt{m}}\sum\limits_{v}\vert rv\rangle,\\
		&\vert e_{2}\rangle=\frac{1}{\sqrt{m}}\sum\limits_{u}\vert ur\rangle,\\
		&\vert e_{3}\rangle=\frac{1}{\sqrt{m(n-1)}}\sum\limits_{u,v}\vert uv\rangle,\\
		&\vert e_{4}\rangle=\frac{1}{\sqrt{m(n-1)}}\sum\limits_{v,u}\vert vu\rangle.
	\end{aligned}
\end{equation}
\begin{figure}[H]
	\centering
	\includegraphics[width=0.4\textwidth]{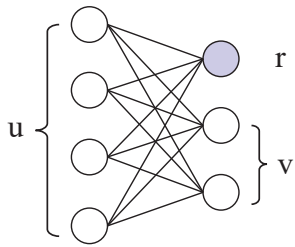}
	\caption{Only the receiver is marked in the second stage. }\label{receiver1}
\end{figure}
The flip-flop shift operator $S_{2}$, the query oracle $Q_{2}(\beta)$ and the coin operator $C_{2}(\alpha)$ can be rewritten as
\begin{equation}
	S_2=\begin{pmatrix}
		0&1&0&0\\
		1&0&0&0\\
		0&0&0&1\\
		0&0&1&0\\
	\end{pmatrix},
	Q_2(\beta)=\begin{pmatrix}
		e^{i\beta}&0&0&0\\
		0&1&0&0\\
		0&0&1&0\\
		0&0&0&1\\
	\end{pmatrix},
\end{equation}
and
\begin{equation}
	C_2(\alpha)=\begin{pmatrix}
		-e^{-i\alpha}&0&0&0\\
		0&\frac{(1-e^{-i\alpha})(1-cos(\omega_2))}{2}-1&\frac{(1-e^{-i\alpha})sin(\omega_2)}{2}&0\\
		0&\frac{(1-e^{-i\alpha})sin(\omega_2)}{2}&\frac{(1-e^{-i\alpha})(1+cos(\omega_2))}{2}-1&0\\
		0&0&0&-e^{-i\alpha}\\
	\end{pmatrix},
\end{equation}
where
$cos(\omega_2)=1-\frac{2}{n}$ and $sin(\omega_2)=\frac{2}{n}\sqrt{n-1}$.

In the second stage, we know 
\begin{equation}
	\vert \psi_{h_{2}}\rangle=S_2C_2(\alpha^{'}_{h_{2}})Q_2(\beta^{'}_{h_{2}})S_2C_2(\alpha^{'}_{h_{2}-1})Q_2(\beta^{'}_{h_{2}-1})...S_2C_2(\alpha^{'}_{1})Q_2(\beta^{'}_{1})\vert \psi_{h_{1}}\rangle.  
\end{equation}
The coin operator $C_{2}(\alpha)$ can be denoted as 
\begin{equation}
	C_2(\alpha)=e^{-\frac{i\alpha}{2}}A_{2}(-\frac{\pi}{2})R_{2}(\alpha)A_{2}(\frac{\pi}{2}),
	\label{C2}
\end{equation}
where
\begin{equation}
	R_{2}(\theta)=-\begin{pmatrix}
		e^{-\frac{i\theta}{2}}&0&0&0\\
		0&e^{\frac{i\theta}{2}}&0&0\\
		0&0&e^{-\frac{i\theta}{2}}&0\\
		0&0&0&e^{-\frac{i\theta}{2}}\\
	\end{pmatrix},
\end{equation}
and 
\begin{equation}
	A_{2}(\theta)=\begin{pmatrix}
		1&0&0&0\\
		0&cos(\frac{\omega_2}{2})&-ie^{i\theta}sin(\frac{\omega_2}{2})&0\\
		0&-ie^{-i\theta}sin(\frac{\omega_2}{2})&cos(\frac{\omega_2}{2})&0\\
		0&0&0&1\\
	\end{pmatrix}.
\end{equation}
The query oracle $Q_{2}(\beta)$ can be denoted as 
\begin{equation}
	Q_2(\beta)=-e^{\frac{i\beta}{2}}S_2R_{2}({\beta})S_2.
	\label{Q2}
\end{equation}
And we find the equation
\begin{equation}
	S_2B_1S_2B_2S_2=B_2S_2B_1,
	\label{S2}
\end{equation}
where
$B_1=\prod_{i=0}^{n_1}D_i,B_2=\prod_{i=0}^{n_2}D_i$, for $D_i\in{A_{2}(\theta_i),R_{2}(\theta_i)}$.

Then by using Eq. (\ref{C2}), Eq. (\ref{Q2}) and  Eq. (\ref{S2}), we have 
\begin{equation}\label{eq3}
	\begin{aligned}
		\vert \psi_{h_{2}}\rangle
		\sim&
		S_2A_{2}(\frac{\pi}{2})R_{2}(\alpha^{'}_{h_{2}})A_{2}(-\frac{\pi}{2})R_{2}(\beta^{'}_{h_{2}-1})A_{2}(\frac{\pi}{2})R_{2}(\alpha^{'}_{h_{2}-2})A_{2}(-\frac{\pi}{2})...\\&R_{2}(\beta^{'}_{5})A_{2}(\frac{\pi}{2})R_{2}(\alpha^{'}_{4})A_{2}(-\frac{\pi}{2})R_{2}(\beta^{'}_{3})A_{2}(\frac{\pi}{2})R_{2}(\alpha^{'}_{2})A_{2}(-\frac{\pi}{2})R_{2}(\beta^{'}_1)\\
		&S_2R_{2}(\beta^{'}_{h_{2}})A_{2}(\frac{\pi}{2})R_{2}(\alpha^{'}_{h_{2}-1})A_{2}(-\frac{\pi}{2})...R_{2}(\beta^{'}_2)A_{2}(\frac{\pi}{2})R_{2}(\alpha^{'}_{1})A_{2}(-\frac{\pi}{2})\vert \psi_{h_{1}}\rangle, 
	\end{aligned}
\end{equation}
where $h_{2}$ is an even integer.

The state $\vert \psi_{h_{1}} \rangle$ can be rewritten as 
$
\vert \psi_{h_{1}}\rangle\approx\vert \Psi \rangle=S_2A_{2}(\frac{\pi}{2})\vert e_{3}\rangle. 
$ Then we eliminate invalid $A_{2}(\theta)$ and $R_{2}(\theta)$. So Eq. (\ref{eq3}) can be simplified to 
\begin{equation}
	\begin{aligned}
		\vert \psi_{h_{2}}\rangle
		\sim&
		S_2A_{2}(\frac{\pi}{2})R_{2}(\alpha^{'}_{h_{2}})A_{2}(-\frac{\pi}{2})R_{2}(\beta^{'}_{h_{2}-1})A_{2}(\frac{\pi}{2})R_{2}(\alpha^{'}_{h_{2}-2})A_{2}(-\frac{\pi}{2})...\\&R_{2}(\beta^{'}_{5})A_{2}(\frac{\pi}{2})R_{2}(\alpha^{'}_{4})A_{2}(-\frac{\pi}{2})R_{2}(\beta^{'}_{3})A_{2}(\frac{\pi}{2})R_{2}(\alpha^{'}_{2})A_{2}(-\frac{\pi}{2})R_{2}(\beta^{'}_1)A_{2}(\frac{\pi}{2})\vert e_{3}\rangle. 
	\end{aligned}
\end{equation}
Then by using
$
	A_{2}(\alpha+\beta)=R_{2}(\beta)A_{2}(\alpha)R_{2}(-\beta)
$
and 
$
	R_{2}(\theta)R_{2}(-\theta)=I
$, we obtain 
\begin{equation}
	\begin{aligned}
		\vert \psi_{h_{2}}\rangle\sim& S_2A_{2}(\frac{\pi}{2})A_{2}(-\frac{\pi}{2}+\alpha^{'}_{h_{2}})A_{2}(\frac{\pi}{2}+\alpha^{'}_{h_{2}}+\beta^{'}_{h_{2}-1})...\\
		&A_{2}(-\frac{\pi}{2}+\alpha^{'}_{h_{2}}+\beta^{'}_{h_{2}-1}+...+\alpha_{2})A_{2}(\frac{\pi}{2}+\alpha^{'}_{h_{2}}+\beta^{'}_{h_{2}-1}+...+\alpha_{2}+\beta_{1})\vert e_{3}\rangle.
	\end{aligned}
\end{equation}
The target state of the second stage is $\vert target\rangle=\frac{1}{\sqrt{m}}\sum\limits_{v}\vert rv\rangle=\vert e_{1}\rangle$. So the fidelity of the second stage can be calculated as follow.
\begin{equation} \label{F2'}
	\begin{aligned}
		F_{2}&=\vert \langle e_{1} \vert S_2A_{2}(\frac{\pi}{2})A_{2}(-\frac{\pi}{2}+\alpha^{'}_{h_{2}})A_{2}(\frac{\pi}{2}+\alpha^{'}_{h_{2}}-\alpha^{'}_{2})...A_{2}(-\frac{\pi}{2}+\alpha^{'}_{h_{2}})A_{2}(\frac{\pi}{2})\vert e_{3}\rangle \vert^{2}
	\end{aligned}
\end{equation}
There exists a set of parameters $\alpha_{i}^{'}$, $\beta_{i}^{'}$, then the value of fidelity $F_{2}$ will greater than or equal to $1-\epsilon_{2}$. It can be shown in theorem \ref{thm3}.
\begin{thm}\label{thm3}
	Let $
	\alpha^{'}_{k}=-\beta^{'}_{h_2+1-k}=2arccot(tan(\frac{k\pi}{h_2+1})\sqrt{1-\gamma_2^2})
	$, for
	$k=2,4,6,...,h_{2},$ where 
	$\gamma_2=\frac{1}{cos(\frac{1}{(h_2+1)}arccos(\frac{1}{\sqrt{\epsilon_2}}))}$
	, and ensure $
	h_2 \ge ln(\frac{2}{\sqrt{\epsilon_2}})\sqrt{n}-1
	$, then the value of fidelity $F_{2} \ge 1-\epsilon_{2}$. 
\end{thm}
\begin{proof}[Proof. ]
	Let $\beta^{'}_{i}=-\alpha^{'}_{h_{2}+1-i}$, for $i=1,3,5,...,h_{2}-1$. So Eq.(\ref{F2'}) can be rewritten as 
	\begin{equation}\label{F2''}
		\begin{aligned}
			F_{2}=\vert \langle e_{1} \vert S_2A_{2}(\zeta_{h_{2}+1})A_{2}(\zeta_{h_{2}})A_{2}(\zeta_{h_{2}-1})...A_{2}(\zeta_{2})A_{2}(\zeta_{1})\vert e_{3}\rangle \vert^{2},
		\end{aligned}
	\end{equation}
	where
	$\zeta_{k+1}-\zeta_{k}=\pi-\alpha^{'}_{k}$ for $k=2,4,6,...,h_{2}$ and $\zeta_{k+1}-\zeta_{k}=-\pi+\alpha^{'}_{h_{2}-k+1}$ for $k=1,3,5,...,h_{2}-1$. 
	
	The formula $S_2A_{2}(\zeta_{h_{2}+1})A_{2}(\zeta_{h_{2}})A_{2}(\zeta_{h_{2}-1})...A_{2}(\zeta_{2})A_{2}(\zeta_{1})\vert e_{3}\rangle$ in Eq. (\ref{F2''}) can be viewed as  
	the operator $S_2A_{2}(\zeta_{h_{2}+1})A_{2}(\zeta_{h_{2}})A_{2}(\zeta_{h_{2}-1})...A_{2}(\zeta_{2})A_{2}(\zeta_{1})$ applied to $\vert e_{3}\rangle$. So it can be divided into two steps as follow.
	\begin{equation}
		\begin{pmatrix}
			0\\0\\1\\0
		\end{pmatrix}
		\xrightarrow[\textcircled{1}]{A_{2}(\zeta_{h_{2}+1})A_{2}(\zeta_{h_{2}})A_{2}(\zeta_{h_{2}-1})...A_{2}(\zeta_{2})A_{2}(\zeta_{1})}
		\begin{pmatrix}
			0\\b_{h_{2}+1}(x)\\c_{h_{2}+1}(x)\\0
		\end{pmatrix}
		\xrightarrow[\textcircled{2}]{S_2}
		\begin{pmatrix}
			b_{h_{2}+1}(x)\\0\\0\\c_{h_{2}+1}(x)
		\end{pmatrix}\nonumber
	\end{equation}
	
	Then after calculations like the proof of the theorem \ref{thm1}, the recurrence formula of $c_{k}(x)$ can be defined by 
	$
	c_0(x)=1,c_1(x)=x$ and for $k=2,3,4,...,h_{2}+1$, 
	\begin{equation}
		c_{k}(x)=x(1+e^{-i(\zeta_{k}-\zeta_{k-1})})c_{k-1}(x)-e^{-i(\zeta_{k}-\zeta_{k-1})}c_{k-2}(x),
	\end{equation}
	with $x=cos(\frac{\omega_2}{2})$.
	
	Let 
	$
	\alpha^{'}_{k}=2arccot(tan(\frac{k\pi}{h_2+1})\sqrt{1-\gamma_2^2})
	$, for
	$k=2,4,6,...,h_{2},$ where 
	$\gamma_2=\frac{1}{cos(\frac{1}{(h_2+1)}arccos(\frac{1}{\sqrt{\epsilon_2}}))}$. So we have $\zeta_{k+1}-\zeta_{k}=(-1)^{k}\pi-2arctan(tan(\frac{k\pi}{h_{2}})\sqrt{1-\gamma_{2}^2})$. 
	By using lemma \ref{lem1}, we obtain 
	\begin{equation}
		c_{h_2+1}(x)=\frac{T_{h_2+1}(\frac{x}{\gamma_2})}{T_{h_2+1}(\frac{1}{\gamma_2})}
		=\sqrt{\epsilon_2}T_{h_2+1}(cos(\frac{1}{(h_2+1)}arccos(\frac{1}{\sqrt{\epsilon_2}}))\sqrt{1-\frac{1}{n}}).
	\end{equation}
	So the fidelity of the second stage can be calculated as follow.
	\begin{equation}
		F_2=1-\vert c_{h_2+1}(x) \vert ^2=1-\epsilon_2 T_{h_2+1}^2(cos(\frac{1}{h_2+1}arccos(\frac{1}{\sqrt{\epsilon_2}}))\sqrt{1-\frac{1}{n}})
	\end{equation}
	Let
	$
	h_2 \ge ln(\frac{2}{\sqrt{\epsilon_2}})\sqrt{n}-1
	$. Similar to the proof of the theorem \ref{thm1}, we have 
	$
	F_2\ge 1-\epsilon_2
	$. 
\end{proof}

Therefore, let $
\alpha^{'}_{k}=-\beta^{'}_{h_2+1-k}=2arccot(tan(\frac{k\pi}{h_2+1})\sqrt{1-\gamma_2^2})
$, for
$k=2,4,6,...,h_{2},$ where 
$\gamma_2=\frac{1}{cos(\frac{1}{(h_2+1)}arccos(\frac{1}{\sqrt{\epsilon_2}}))}$
, and ensure $
h_2 \ge ln(\frac{2}{\sqrt{\epsilon_2}})\sqrt{n}-1
$, the uniform superposition state of the vertices on the other side of the sender will be transferred to the receiver with the fidelity of at least $1-\epsilon_{2}$.
\subsection{The fidelity of the quantum state transfer algorithm}\label{4.3}
Since the sender and receiver are in different partitions of the complete bipartite graph(shown in Fig. \ref{senderAndReceiver1}),  
the analysis of the algorithm can be simplified in an invariant subspace with the orthogonal
basis $\{\vert e_{1}\rangle, \vert e_{2}\rangle, \vert e_{3}\rangle, \vert e_{4}\rangle, \vert e_{5}\rangle, \vert e_{6}\rangle, \vert e_{7}\rangle, \vert e_{8}\rangle \}$ given below. The orthogonal basis is only used in \ref{4.3}.
\begin{equation}
	\begin{aligned}
		&\vert e_{1}\rangle=\vert sr\rangle
		,\\
		&\vert e_{2}\rangle=\vert rs\rangle
		,\\
		&\vert e_{3}\rangle=\frac{1}{\sqrt{n-1}}\sum\limits_{v}\vert sv\rangle
		,\\
		&\vert e_{4}\rangle=\frac{1}{\sqrt{n-1}}\sum\limits_{v}\vert vs\rangle
		,\\
		&\vert e_{5}\rangle=\frac{1}{\sqrt{m-1}}\sum\limits_{u}\vert ur\rangle,\\
		&\vert e_{6}\rangle=\frac{1}{\sqrt{m-1}}\sum\limits_{u}\vert ru\rangle.\\ 
		&\vert e_{7}\rangle=\frac{1}{\sqrt{(m-1)(n-1)}}\sum\limits_{u,v}\vert uv\rangle,\\
		&\vert e_{8}\rangle=\frac{1}{\sqrt{(m-1)(n-1)}}\sum\limits_{v,u}\vert vu\rangle. 
	\end{aligned}.
\end{equation}
\begin{figure}[H]
	\centering
	\includegraphics[width=0.4\textwidth]{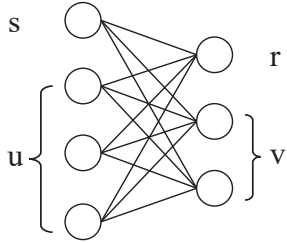}
	\caption{The sender and the receiver are in different partitions. }
	\label{senderAndReceiver1}
\end{figure}
From the analysis of the first stage in \ref{3.1}, we can obtain 
$\vert \psi_{h_{1}}\rangle \sim (0,\sqrt{\frac{m-1}{m}}b_{h_{1}}(x)+\frac{1}{\sqrt{m}}c_{h_{1}}(x),-\frac{1}{\sqrt{m}}b_{h_{1}}(x)+\sqrt{\frac{m-1}{m}}c_{h_{1}}(x),0)^T$.
So in the new basis, the state $\vert \psi_{h_{1}} \rangle$ can be rewritten as 
\begin{equation}
	\vert \psi_{h_{1}}\rangle \sim t_{1}\vert \Psi\rangle+t_{2}\vert e_{2}\rangle+\sqrt{n-1}t_{2}\vert e_{4}\rangle,
\end{equation}
where $t_{1}=c_{h_{1}}(x)-\frac{1}{\sqrt{m-1}}b_{h_{1}}(x)$,  $t_{2}=\sqrt{\frac{m}{n(m-1)}}b_{h_{1}}(x)$ and $\vert \Psi \rangle =\frac{1}{\sqrt{m}}\vert e_{2}\rangle +\frac{1}{\sqrt{m}}\vert e_{4}\rangle +\frac{\sqrt{m-2}}{\sqrt{m}}\vert e_{6}\rangle.$ $\vert \Psi \rangle$ denotes the target state of the first stage. 

So in the second stage, we have 
\begin{equation}
	\vert \psi_{h_{2}}\rangle \sim t_{1}U_{2}\vert \Psi\rangle+t_{2}U_{2}\vert e_{2}\rangle+\sqrt{n-1}t_{2}U_{2}\vert e_{4}\rangle,
\end{equation}
where $U_{2}$ denotes the evolution operators of the second stage. 

Let $t_{1}U_{2}\vert \Psi\rangle=(0,f_{2},0,f_{4},0,f_{6},0,f_{8})^{T}$, $t_{2}U_{2}\vert e_{2}\rangle=(0,g_{2},0,g_{4},0,g_{6},0,g_{8})^{T}$ and $\sqrt{n-1 }t_{2}U_{2}\vert e_{4}\rangle=(0,l_{2},0,l_{4},0,l_{6},0,l_{8})^{T}$. So we can obtain 
\begin{equation}\label{84}
	\begin{cases}
		\vert f_{2} \vert^{2} +\vert f_{4} \vert^{2}+\vert f_{6} \vert^{2}+\vert f_{8} \vert^{2}=\vert t_{1} \vert ^{2},\\
		\vert g_{2} \vert^{2} +\vert g_{4} \vert^{2}+\vert g_{6} \vert^{2}+\vert g_{8} \vert^{2}=\vert t_{2} \vert ^{2},\\
		\vert l_{2} \vert^{2} +\vert l_{4} \vert^{2}+\vert l_{6} \vert^{2}+\vert l_{8} \vert^{2}=(n-1)\vert t_{2} \vert ^{2},\\
		\vert f_{2}+g_{2}+l_{2}\vert^{2}+\vert f_{4}+g_{4}+l_{4}\vert^{2}+\vert f_{6}+g_{6}+l_{6}\vert^{2}+\vert f_{8}+g_{8}+l_{8}\vert^{2}=1.
	\end{cases}
\end{equation}

The target state of the algorithm is $ \frac{1}{\sqrt{m}}\vert e_{2}\rangle+\frac{\sqrt{m-1}}{\sqrt{m}}\vert e_{6}\rangle. $
So the fidelity of the algorithm can be denoted as 
\begin{equation}\label{85}
	F=\bm{\big\vert} \frac{1}{\sqrt{m}}(f_{2} + g_{2}+ l_{2})+
	\frac{\sqrt{m-1}}{\sqrt{m}}(f_{6} + g_{6}+ l_{6})\bm{\big\vert} ^{2}.
\end{equation}
From \ref{4.2}, we know $f_{6}=\sqrt{m-1}f_{2}$. 
So we can obtain 
\begin{equation}
	\begin{aligned}
		F& \ge \vert f_{2}+g_{2}+l_{2}\vert^{2} + \vert f_{6}+g_{6}+l_{6}\vert^{2} -\vert g_{2} \vert^{2} -\vert l_{2} \vert^{2}-\vert g_{6} \vert^{2}-\vert l_{6} \vert^{2} -2\vert g_{2}\vert \vert l_{2}\vert -2\vert g_{6}\vert \vert l_{6}\vert.
	\end{aligned}
\end{equation}
Then by using Eq. (\ref{84}), we have  
\begin{equation}
	\begin{aligned}
		F\ge 1- \vert f_{4}+g_{4}+l_{4}\vert^{2} - \vert f_{8}+g_{8}+l_{8}\vert^{2} 
		-\vert g_{2} \vert^{2} -\vert l_{2} \vert^{2}-\vert g_{6} \vert^{2}-\vert l_{6} \vert^{2} 
		-2\vert g_{2}\vert \vert l_{2}\vert -2\vert g_{6}\vert \vert l_{6}\vert.
	\end{aligned}
\end{equation}
By using $\vert x+y \vert \leq \vert \vert x \vert +\vert y \vert \vert$, we obtain  
\begin{equation}
	\begin{aligned}
		F\ge& 1
		-
		(\vert f_{4}\vert ^{2}
		+\vert f_{8}\vert ^{2})
		-(\vert g_{2} \vert^{2}
		+\vert g_{4}\vert^{2}
		+\vert g_{6} \vert^{2}
		+\vert g_{8}\vert^{2}
		+\vert l_{2} \vert^{2}
		+\vert l_{4}\vert^{2}
		+\vert l_{6} \vert^{2} 
		+\vert l_{8}\vert^{2})
		\\
		&
		-2(\vert f_{4}\vert \vert g_{4}\vert
		+\vert f_{8}\vert \vert g_{8}\vert 
		+\vert f_{4}\vert \vert l_{4}\vert
		+\vert f_{8}\vert \vert l_{8}\vert
		+\vert g_{2}\vert \vert l_{2}\vert 
		+\vert g_{4}\vert \vert l_{4}\vert
		+\vert g_{6}\vert \vert l_{6}\vert
		+\vert g_{8}\vert \vert l_{8}\vert)
		.
	\end{aligned}
\end{equation}

From \ref{4.2}, 
we know $\vert f_{4}\vert ^{2}+\vert f_{8}\vert ^{2}\leq \vert t_{1} \vert^{2} \epsilon_{2} \textless \epsilon_{2}$.
From Eq. (\ref{84}), we have $\vert g_{2} \vert^{2}
+\vert g_{4}\vert^{2}
+\vert g_{6} \vert^{2}
+\vert g_{8}\vert^{2}
+\vert l_{2} \vert^{2}
+\vert l_{4}\vert^{2}
+\vert l_{6} \vert^{2} 
+\vert l_{8}\vert^{2}= n\vert t_{2} \vert^{2}\leq 2\epsilon_{1}$. 

We know $\vert f_{4}\vert 
\vert g_{4}\vert+
\vert f_{8}\vert 
\vert g_{8}\vert 
\leq \sqrt{(\vert f_{4}\vert^2+\vert f_{8}\vert^2)
	(\vert g_{4}\vert^2+\vert g_{8}\vert^2)
}
\textless \sqrt{\epsilon_{1}\epsilon_{2}}$.
We also have $
\vert f_{4}\vert 
\vert l_{4}\vert+
\vert f_{8}\vert 
\vert l_{8}\vert 
\leq 
\sqrt{(\vert f_{4}\vert^2+\vert f_{8}\vert^2)
	(\vert l_{4}\vert^2+\vert l_{8}\vert^2)
}\textless \sqrt{2\epsilon_{1}\epsilon_{2}}$. 
And we have
$\vert g_{2}\vert \vert l_{2}\vert 
+\vert g_{4}\vert \vert l_{4}\vert
+\vert g_{6}\vert \vert l_{6}\vert
+\vert g_{8}\vert \vert l_{8}\vert
\leq
\sqrt{(\vert g_{2} \vert^{2} +\vert g_{4} \vert^{2}+\vert g_{6} \vert^{2}+\vert g_{8} \vert^{2})
	(\vert l_{2} \vert^{2} +\vert l_{4} \vert^{2}+\vert l_{6} \vert^{2}+\vert l_{8} \vert^{2})
}\textless \sqrt{2}\epsilon_{1}
$.

So we obtain 
\begin{equation}\label{88}
	F \textgreater 1-(2+2\sqrt{2})\epsilon_{1}-\epsilon_{2}-(2+2\sqrt{2})\sqrt{\epsilon_{1}\epsilon_{2}}.
\end{equation}
From Eq. (\ref{88}),  we know that the fidelity will be close to 1 when $\epsilon_{1}$ and $\epsilon_{2}$ are small.
For instance, let $\epsilon_{1}$ be $0.01$ and $\epsilon_{2}$ be $0.01$. From Eq. (\ref{88}), we know the fidelity will be greater than $0.89$ regardless of the value of $m$ and $n$.
The simulation results of the algorithm are shown in Fig. \ref{fidelity2}. 
The fidelity is bigger than 0.98 at a certain range of $m$ and $n$ when $\epsilon_{1}=0.01$ and $\epsilon_{2}=0.01$. 
It further verifies that the quantum state transfer algorithm can achieve high fidelity. 
\begin{figure}[H]
	\centering
	\includegraphics[width=0.6\textwidth]{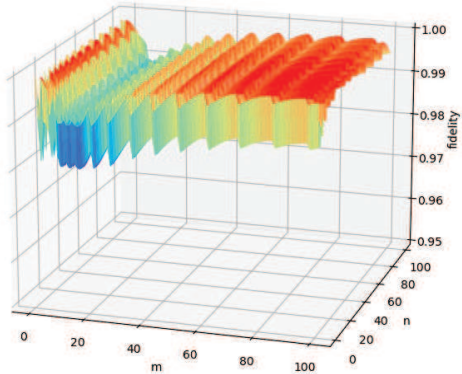}
	\caption{The fidelity of the quantum state transfer algorithm with $\epsilon_{1}=0.01$ and $\epsilon_{2}=0.01$ when the sender and receiver are in different partitions. }
	\label{fidelity2}
\end{figure}
\section{Conclusions}\label{sec5}
In this paper, we propose a high-fidelity quantum state transfer algorithm on the complete bipartite graph. 
The algorithm is divided into two stages. 
The first stage is to transfer the initial state to the uniform superposition state of the vertices on the other side of the sender. 
The second stage is to transfer the uniform superposition state of the vertices on the other side of the sender to the target state.
The two stages are both achieved by using the generalized Grover walks with one marked vertex.
The coin operators of the generalized Grover walks and the query oracles are parametric unitary matrices that changed with time. 

Through analysis, it is found that in the first stage, the initial state is transferred to the uniform superposition state of the vertex on the other side of the sender with the fidelity of at least $1 - \epsilon_{1} $. 
In the second stage, the uniform superposition state of the vertices on the other side of the sender is transferred to the target state with the fidelity of at least $1-\epsilon_{2}$. 
We prove that the fidelity of the algorithm is greater than $ 1-2\epsilon_{1}-\epsilon_{2}-2\sqrt{2}\sqrt{\epsilon_{1}\epsilon_{2}}$ or $1-(2+2\sqrt{2})\epsilon_{1}-\epsilon_{2}-(2+2\sqrt{2})\sqrt{\epsilon_{1}\epsilon_{2}}$  
when the sender and receiver are in the same partition or different partitions.
$\epsilon_{1}$ and $\epsilon_{2}$ are chosen from $(0,1]$. 
When $\epsilon_{1}$ and $\epsilon_{2}$ are small, the fidelity of the algorithm will be close to 1. 

Consequently, the algorithm can achieve high-fidelity quantum state transfer when the sender and receiver are located in the same partition or different partitions of the complete bipartite graph. 
Moreover, the algorithm can achieve high-fidelity quantum state transfer on complete bipartite graphs of various sizes.
Compared to the previous algorithms, the advantage of the algorithm is it works in any case because high-fidelity quantum state transfer can be achieved by adjusting the parameters of the coin operators and the query oracles.
The algorithm provides a novel method for achieving high-fidelity quantum state transfer on the complete bipartite graph, which will offer potential applications for quantum information processing.
\section*{Data Availability Statement}
Data sharing is not applicable to this article as no datasets were generated or analysed during the current study.
\section*{Conflict of interest statement}
The authors do not have any possible conflicts of interest.
\section*{Acknowledgements}
This work is supported by NSFC (Grant Nos. 61901218, 62071015) and the National Key Research and Development Program of China (Grant No.2020YFB1005504).

\bibliography{sn-bibliography}


\begin{thebibliography}{27}
\ifx \bisbn   \undefined \def \bisbn  #1{ISBN #1}\fi
\ifx \binits  \undefined \def \binits#1{#1}\fi
\ifx \bauthor  \undefined \def \bauthor#1{#1}\fi
\ifx \batitle  \undefined \def \batitle#1{#1}\fi
\ifx \bjtitle  \undefined \def \bjtitle#1{#1}\fi
\ifx \bvolume  \undefined \def \bvolume#1{\textbf{#1}}\fi
\ifx \byear  \undefined \def \byear#1{#1}\fi
\ifx \bissue  \undefined \def \bissue#1{#1}\fi
\ifx \bfpage  \undefined \def \bfpage#1{#1}\fi
\ifx \blpage  \undefined \def \blpage #1{#1}\fi
\ifx \burl  \undefined \def \burl#1{\textsf{#1}}\fi
\ifx \doiurl  \undefined \def \doiurl#1{\url{https://doi.org/#1}}\fi
\ifx \betal  \undefined \def \betal{\textit{et al.}}\fi
\ifx \binstitute  \undefined \def \binstitute#1{#1}\fi
\ifx \binstitutionaled  \undefined \def \binstitutionaled#1{#1}\fi
\ifx \bctitle  \undefined \def \bctitle#1{#1}\fi
\ifx \beditor  \undefined \def \beditor#1{#1}\fi
\ifx \bpublisher  \undefined \def \bpublisher#1{#1}\fi
\ifx \bbtitle  \undefined \def \bbtitle#1{#1}\fi
\ifx \bedition  \undefined \def \bedition#1{#1}\fi
\ifx \bseriesno  \undefined \def \bseriesno#1{#1}\fi
\ifx \blocation  \undefined \def \blocation#1{#1}\fi
\ifx \bsertitle  \undefined \def \bsertitle#1{#1}\fi
\ifx \bsnm \undefined \def \bsnm#1{#1}\fi
\ifx \bsuffix \undefined \def \bsuffix#1{#1}\fi
\ifx \bparticle \undefined \def \bparticle#1{#1}\fi
\ifx \barticle \undefined \def \barticle#1{#1}\fi
\bibcommenthead
\ifx \bconfdate \undefined \def \bconfdate #1{#1}\fi
\ifx \botherref \undefined \def \botherref #1{#1}\fi
\ifx \url \undefined \def \url#1{\textsf{#1}}\fi
\ifx \bchapter \undefined \def \bchapter#1{#1}\fi
\ifx \bbook \undefined \def \bbook#1{#1}\fi
\ifx \bcomment \undefined \def \bcomment#1{#1}\fi
\ifx \oauthor \undefined \def \oauthor#1{#1}\fi
\ifx \citeauthoryear \undefined \def \citeauthoryear#1{#1}\fi
\ifx \endbibitem  \undefined \def \endbibitem {}\fi
\ifx \bconflocation  \undefined \def \bconflocation#1{#1}\fi
\ifx \arxivurl  \undefined \def \arxivurl#1{\textsf{#1}}\fi
\csname PreBibitemsHook\endcsname

\bibitem{kadian2021quantum}
\begin{barticle}
\bauthor{\bsnm{Kadian}, \binits{K.}},
\bauthor{\bsnm{Garhwal}, \binits{S.}},
\bauthor{\bsnm{Kumar}, \binits{A.}}:
\batitle{Quantum walk and its application domains: A systematic review}.
\bjtitle{Computer Science Review}
\bvolume{41},
\bfpage{100419}
(\byear{2021})
\end{barticle}
\endbibitem

\bibitem{venegas2012quantum}
\begin{barticle}
\bauthor{\bsnm{Venegas-Andraca}, \binits{S.E.}}:
\batitle{Quantum walks: a comprehensive review}.
\bjtitle{Quantum Information Processing}
\bvolume{11}(\bissue{5}),
\bfpage{1015}--\blpage{1106}
(\byear{2012})
\end{barticle}
\endbibitem

\bibitem{aharonov1993quantum}
\begin{barticle}
\bauthor{\bsnm{Aharonov}, \binits{Y.}},
\bauthor{\bsnm{Davidovich}, \binits{L.}},
\bauthor{\bsnm{Zagury}, \binits{N.}}:
\batitle{Quantum random walks}.
\bjtitle{Physical Review A}
\bvolume{48}(\bissue{2}),
\bfpage{1687}
(\byear{1993})
\end{barticle}
\endbibitem

\bibitem{childs2009universal}
\begin{barticle}
\bauthor{\bsnm{Childs}, \binits{A.M.}}:
\batitle{Universal computation by quantum walk}.
\bjtitle{Physical review letters}
\bvolume{102}(\bissue{18}),
\bfpage{180501}
(\byear{2009})
\end{barticle}
\endbibitem

\bibitem{lovett2010universal}
\begin{barticle}
\bauthor{\bsnm{Lovett}, \binits{N.B.}},
\bauthor{\bsnm{Cooper}, \binits{S.}},
\bauthor{\bsnm{Everitt}, \binits{M.}},
\bauthor{\bsnm{Trevers}, \binits{M.}},
\bauthor{\bsnm{Kendon}, \binits{V.}}:
\batitle{Universal quantum computation using the discrete-time quantum walk}.
\bjtitle{Physical Review A}
\bvolume{81}(\bissue{4}),
\bfpage{042330}
(\byear{2010})
\end{barticle}
\endbibitem

\bibitem{reitzner2009quantum}
\begin{barticle}
\bauthor{\bsnm{Reitzner}, \binits{D.}},
\bauthor{\bsnm{Hillery}, \binits{M.}},
\bauthor{\bsnm{Feldman}, \binits{E.}},
\bauthor{\bsnm{Bu{\v{z}}ek}, \binits{V.}}:
\batitle{Quantum searches on highly symmetric graphs}.
\bjtitle{Physical Review A}
\bvolume{79}(\bissue{1}),
\bfpage{012323}
(\byear{2009})
\end{barticle}
\endbibitem

\bibitem{rhodes2019quantum}
\begin{barticle}
\bauthor{\bsnm{Rhodes}, \binits{M.L.}},
\bauthor{\bsnm{Wong}, \binits{T.G.}}:
\batitle{Quantum walk search on the complete bipartite graph}.
\bjtitle{Physical Review A}
\bvolume{99}(\bissue{3}),
\bfpage{032301}
(\byear{2019})
\end{barticle}
\endbibitem

\bibitem{yalccinkaya2015qubit}
\begin{barticle}
\bauthor{\bsnm{Yal{\c{c}}{\i}nkaya}, \binits{{\. I}.}},
\bauthor{\bsnm{Gedik}, \binits{Z.}}:
\batitle{Qubit state transfer via discrete-time quantum walks}.
\bjtitle{Journal of Physics A: Mathematical and Theoretical}
\bvolume{48}(\bissue{22}),
\bfpage{225302}
(\byear{2015})
\end{barticle}
\endbibitem

\bibitem{zhan2014perfect}
\begin{barticle}
\bauthor{\bsnm{Zhan}, \binits{X.}},
\bauthor{\bsnm{Qin}, \binits{H.}},
\bauthor{\bsnm{Bian}, \binits{Z.-h.}},
\bauthor{\bsnm{Li}, \binits{J.}},
\bauthor{\bsnm{Xue}, \binits{P.}}:
\batitle{Perfect state transfer and efficient quantum routing: A discrete-time
  quantum-walk approach}.
\bjtitle{Physical Review A}
\bvolume{90}(\bissue{1}),
\bfpage{012331}
(\byear{2014})
\end{barticle}
\endbibitem

\bibitem{li2022controlled}
\begin{botherref}
\oauthor{\bsnm{Li}, \binits{D.}},
\oauthor{\bsnm{Ding}, \binits{P.}},
\oauthor{\bsnm{Zhou}, \binits{Y.}},
\oauthor{\bsnm{Yang}, \binits{Y.}}:
Controlled alternate quantum walk based block hash function.
arXiv preprint arXiv:2205.05983
(2022)
\end{botherref}
\endbibitem

\bibitem{li2013discrete}
\begin{barticle}
\bauthor{\bsnm{Li}, \binits{D.}},
\bauthor{\bsnm{Zhang}, \binits{J.}},
\bauthor{\bsnm{Guo}, \binits{F.-Z.}},
\bauthor{\bsnm{Huang}, \binits{W.}},
\bauthor{\bsnm{Wen}, \binits{Q.-Y.}},
\bauthor{\bsnm{Chen}, \binits{H.}}:
\batitle{Discrete-time interacting quantum walks and quantum hash schemes}.
\bjtitle{Quantum information processing}
\bvolume{12}(\bissue{3}),
\bfpage{1501}--\blpage{1513}
(\byear{2013})
\end{barticle}
\endbibitem

\bibitem{ambainis2007quantum}
\begin{barticle}
\bauthor{\bsnm{Ambainis}, \binits{A.}}:
\batitle{Quantum walk algorithm for element distinctness}.
\bjtitle{SIAM Journal on Computing}
\bvolume{37}(\bissue{1}),
\bfpage{210}--\blpage{239}
(\byear{2007})
\end{barticle}
\endbibitem

\bibitem{magniez2007quantum}
\begin{barticle}
\bauthor{\bsnm{Magniez}, \binits{F.}},
\bauthor{\bsnm{Santha}, \binits{M.}},
\bauthor{\bsnm{Szegedy}, \binits{M.}}:
\batitle{Quantum algorithms for the triangle problem}.
\bjtitle{SIAM Journal on Computing}
\bvolume{37}(\bissue{2}),
\bfpage{413}--\blpage{424}
(\byear{2007})
\end{barticle}
\endbibitem

\bibitem{reitzner2017finding}
\begin{barticle}
\bauthor{\bsnm{Reitzner}, \binits{D.}},
\bauthor{\bsnm{Hillery}, \binits{M.}},
\bauthor{\bsnm{Koch}, \binits{D.}}:
\batitle{Finding paths with quantum walks or quantum walking through a maze}.
\bjtitle{Physical Review A}
\bvolume{96}(\bissue{3}),
\bfpage{032323}
(\byear{2017})
\end{barticle}
\endbibitem

\bibitem{wang2019controlled}
\begin{barticle}
\bauthor{\bsnm{Wang}, \binits{Y.}},
\bauthor{\bsnm{Wu}, \binits{S.}},
\bauthor{\bsnm{Wang}, \binits{W.}}:
\batitle{Controlled quantum search on structured databases}.
\bjtitle{Physical Review Research}
\bvolume{1}(\bissue{3}),
\bfpage{033016}
(\byear{2019})
\end{barticle}
\endbibitem

\bibitem{childs2004spatial}
\begin{barticle}
\bauthor{\bsnm{Childs}, \binits{A.M.}},
\bauthor{\bsnm{Goldstone}, \binits{J.}}:
\batitle{Spatial search by quantum walk}.
\bjtitle{Physical Review A}
\bvolume{70}(\bissue{2}),
\bfpage{022314}
(\byear{2004})
\end{barticle}
\endbibitem

\bibitem{philipp2016continuous}
\begin{barticle}
\bauthor{\bsnm{Philipp}, \binits{P.}},
\bauthor{\bsnm{Tarrataca}, \binits{L.}},
\bauthor{\bsnm{Boettcher}, \binits{S.}}:
\batitle{Continuous-time quantum search on balanced trees}.
\bjtitle{Physical Review A}
\bvolume{93}(\bissue{3}),
\bfpage{032305}
(\byear{2016})
\end{barticle}
\endbibitem

\bibitem{divincenzo2000physical}
\begin{barticle}
\bauthor{\bsnm{DiVincenzo}, \binits{D.P.}}:
\batitle{The physical implementation of quantum computation}.
\bjtitle{Fortschritte der Physik: Progress of Physics}
\bvolume{48}(\bissue{9-11}),
\bfpage{771}--\blpage{783}
(\byear{2000})
\end{barticle}
\endbibitem

\bibitem{shang2019quantum}
\begin{barticle}
\bauthor{\bsnm{Shang}, \binits{Y.}},
\bauthor{\bsnm{Wang}, \binits{Y.}},
\bauthor{\bsnm{Li}, \binits{M.}},
\bauthor{\bsnm{Lu}, \binits{R.}}:
\batitle{Quantum communication protocols by quantum walks with two coins}.
\bjtitle{EPL (Europhysics Letters)}
\bvolume{124}(\bissue{6}),
\bfpage{60009}
(\byear{2019})
\end{barticle}
\endbibitem

\bibitem{chen2019quantum}
\begin{barticle}
\bauthor{\bsnm{Chen}, \binits{X.-B.}},
\bauthor{\bsnm{Wang}, \binits{Y.-L.}},
\bauthor{\bsnm{Xu}, \binits{G.}},
\bauthor{\bsnm{Yang}, \binits{Y.-X.}}:
\batitle{Quantum network communication with a novel discrete-time quantum
  walk}.
\bjtitle{Ieee Access}
\bvolume{7},
\bfpage{13634}--\blpage{13642}
(\byear{2019})
\end{barticle}
\endbibitem

\bibitem{vstefavnak2016perfect}
\begin{barticle}
\bauthor{\bsnm{{\v{S}}tefa{\v{n}}{\'a}k}, \binits{M.}},
\bauthor{\bsnm{Skoup{\`y}}, \binits{S.}}:
\batitle{Perfect state transfer by means of discrete-time quantum walk search
  algorithms on highly symmetric graphs}.
\bjtitle{Physical Review A}
\bvolume{94}(\bissue{2}),
\bfpage{022301}
(\byear{2016})
\end{barticle}
\endbibitem

\bibitem{vstefavnak2017perfect}
\begin{barticle}
\bauthor{\bsnm{{\v{S}}tefa{\v{n}}{\'a}k}, \binits{M.}},
\bauthor{\bsnm{Skoup{\`y}}, \binits{S.}}:
\batitle{Perfect state transfer by means of discrete-time quantum walk on
  complete bipartite graphs}.
\bjtitle{Quantum Information Processing}
\bvolume{16}(\bissue{3}),
\bfpage{1}--\blpage{14}
(\byear{2017})
\end{barticle}
\endbibitem

\bibitem{skoupy2021quantum}
\begin{barticle}
\bauthor{\bsnm{Skoup{\`y}}, \binits{S.}},
\bauthor{\bsnm{{\v{S}}tefa{\v{n}}{\'a}k}, \binits{M.}}:
\batitle{Quantum-walk-based state-transfer algorithms on the complete m-partite
  graph}.
\bjtitle{Physical Review A}
\bvolume{103}(\bissue{4}),
\bfpage{042222}
(\byear{2021})
\end{barticle}
\endbibitem

\bibitem{zhan2019infinite}
\begin{barticle}
\bauthor{\bsnm{Zhan}, \binits{H.}}:
\batitle{An infinite family of circulant graphs with perfect state transfer in
  discrete quantum walks}.
\bjtitle{Quantum Information Processing}
\bvolume{18}(\bissue{12}),
\bfpage{1}--\blpage{26}
(\byear{2019})
\end{barticle}
\endbibitem

\bibitem{santos2022quantum}
\begin{barticle}
\bauthor{\bsnm{Santos}, \binits{R.A.}}:
\batitle{Quantum state transfer on the complete bipartite graph}.
\bjtitle{Journal of Physics A: Mathematical and Theoretical}
\bvolume{55}(\bissue{12}),
\bfpage{125301}
(\byear{2022})
\end{barticle}
\endbibitem

\bibitem{xu2022robust}
\begin{barticle}
\bauthor{\bsnm{Xu}, \binits{Y.}},
\bauthor{\bsnm{Zhang}, \binits{D.}},
\bauthor{\bsnm{Li}, \binits{L.}}:
\batitle{Robust quantum walk search without knowing the number of marked
  vertices}.
\bjtitle{Physical Review A}
\bvolume{106}(\bissue{5}),
\bfpage{052207}
(\byear{2022})
\end{barticle}
\endbibitem

\bibitem{yoder2014fixed}
\begin{barticle}
\bauthor{\bsnm{Yoder}, \binits{T.J.}},
\bauthor{\bsnm{Low}, \binits{G.H.}},
\bauthor{\bsnm{Chuang}, \binits{I.L.}}:
\batitle{Fixed-point quantum search with an optimal number of queries}.
\bjtitle{Physical review letters}
\bvolume{113}(\bissue{21}),
\bfpage{210501}
(\byear{2014})
\end{barticle}
\endbibitem

\end{thebibliography}

\end{document}